\documentclass[11pt,letterpaper]{article}
\pdfoutput=1

\usepackage{latexsym}
\usepackage[english]{babel}
\usepackage{graphicx}
\usepackage{times}
\usepackage{mathtools}
\usepackage{mathrsfs}
\usepackage{float}
\usepackage{wrapfig}
\usepackage{paralist}
\usepackage{amssymb}
\usepackage{amsthm}
\usepackage{amsmath}
\usepackage{amsfonts}
\usepackage{xspace}
\usepackage[paperwidth=8.5in, paperheight=11in, margin=1in]{geometry}
\usepackage{changepage}
\usepackage{appendix}
\usepackage[dvipsnames,usenames]{xcolor}
\usepackage[ruled,vlined,linesnumbered]{algorithm2e}
\usepackage{enumerate}
\usepackage{tabularx}
\usepackage{textcomp}
\usepackage{color}
\usepackage{subcaption}
\usepackage{caption}
\definecolor{darkgreen}{rgb}{0,0.5,0}
\usepackage{hyperref}
\hypersetup{
    unicode=false,          
    colorlinks=true,        
    linkcolor=red,          
    citecolor=darkgreen,        
    filecolor=magenta,      
    urlcolor=blue           
}
\usepackage[capitalize, nameinlink]{cleveref}

\newtheorem{theorem}{Theorem}[section]
\newtheorem{lemma}[theorem]{Lemma}
\newtheorem{corollary}[theorem]{Corollary}
\newtheorem{definition}{Definition}[section]
\newtheorem{remark}{Remark}[section]

\newcommand{\pushline}{\Indp}
\newcommand{\popline}{\Indm}
\let\oldnl\nl
\newcommand{\nonl}{\renewcommand{\nl}{\let\nl\oldnl}}

\providecommand{\keywords}[1]{\textbf{Keywords:} #1}

\newcommand{\E}{\ensuremath{\mathbb{E}}}
\newcommand{\set}[1]{\left\{#1\right\}}
\newcommand{\bigO}[1]{\ensuremath{\operatorname{O}\bigl(#1\bigr)}}
\newcommand{\para}[1]{\medskip\noindent\textbf{#1}}

\newcommand{\Delay}{\ensuremath{\mathit{C}}}
\newcommand{\delay}{\ensuremath{\mathit{c}}}
\newcommand{\Latency}{L}
\newcommand{\latency}{\ell}
\newcommand{\Weight}{W}
\newcommand{\weight}{w}
\newcommand{\R}{\ensuremath{\mathcal{R}}}
\newcommand{\arrow}{\textsc{Arrow}}
\newcommand{\opt}{\textsc{Opt}}
\newcommand{\gnn}{\textsc{Gnn}}
\newcommand{\alg}{\textsc{Alg}}
\newcommand{\grd}{\textsc{Grd}}
\newcommand{\dsms}{DSMS}
\newcommand{\pointer}{\textit{links}}
\newcommand{\forest}{\ensuremath{\mathcal{F}}}
\newcommand{\src}{\ensuremath{\mathit{src}}}
\newcommand{\des}{\ensuremath{\mathit{des}}}

\title{Concurrent Distributed Serving with Mobile Servers\footnote{This work is supported by the Deutsche Forschungsgemeinschaft (DFG), under grant DFG
TU 221/6-3. A shorter version of this paper is to appear in the proceedings of ISAAC 2019.}}
\author{
Abdolhamid Ghodselahi\thanks{Hamburg University of Technology, Germany. {\small \texttt{\{abdolhamid.ghodselahi,turau\}@tuhh.de}}}
\and
Fabian Kuhn\thanks{University of Freiburg, Germany. {\small \texttt{kuhn@cs.uni-freiburg.de}}}
\and
Volker Turau\footnotemark[2]}
\date{}

\begin{document}

\maketitle

\begin{abstract}
This paper introduces a new resource allocation problem in distributed
computing called \textit{distributed serving with mobile servers
  (\dsms)}. In \dsms, there are $k$ identical mobile servers residing
at the processors of a network. At arbitrary points of time, any
subset of processors can invoke one or more requests. To serve a request, one of
the servers must move to the processor that invoked the request.
Resource allocation is performed in a distributed manner since only
the processor that invoked the request initially knows about it. All
processors cooperate by passing messages to achieve correct resource
allocation. They do this with the goal to minimize the communication
cost.

Routing servers in large-scale distributed systems requires a scalable location service.
We introduce the distributed protocol \gnn\ that solves the \dsms\
problem on \textit{overlay trees}. We prove that \gnn\ is
starvation-free and correctly integrates locating the servers and
synchronizing the concurrent access to servers despite asynchrony,
even when the requests are invoked over time. Further, we analyze
\gnn\ for ``one-shot'' executions, i.e., all requests are invoked
simultaneously. We prove that when running \gnn\ on top of a special
family of tree topologies---known as \textit{hierarchically
  well-separated trees (HSTs)}---we obtain a randomized distributed
protocol with an expected competitive ratio of $\bigO{\log n}$ on
general network topologies with $n$ processors. From a technical point
of view, our main result is that \gnn\ optimally solves the \dsms\
problem on HSTs for one-shot executions, even if communication is
asynchronous. Further, we present a lower bound of
$\Omega(\max\{k, \log n/\log\log n\})$ on the competitive ratio for
\dsms. The lower bound even holds when communication is synchronous
and requests are invoked sequentially.
\end{abstract}

\section{Introduction}
\label{sec:intro}

Consider the following family of online resource allocation problems.
We are given a metric space with $n$ points\label{no:noNodes}.
Initially, a set of
$k \geq 1$\label{no:noServers}\footnote{\Cref{ta:notations} provides
an index for the essential notations used throughout the paper.}
identical mobile servers are residing at different points of the
metric space. Requests arrive over time in an online fashion, that is,
one or several requests can arrive at any point of time. A request
needs to be served by a server at the requesting point sometime after
its arrival. The goal is to provide a schedule for
serving all requests. This
abstract problem lies at the heart of many centralized and distributed
online applications in industrial planning, operating systems, content
distribution in networks, and scheduling
\cite{awerbuch1995,bartal1992competitive,bartal1992distributed,herlihy2001competitive,raymond1989tree}.
Each concrete problem of this family is characterized by a cost
function. We study this abstract problem in distributed computing and
call it the distributed serving with mobile servers (\dsms) problem. A
distributed protocol \alg\ that solves the \dsms\ problem must compute
a schedule for each server consisting of a queue of requests such that
consecutive requests are successively served, and all requests are
served. The $k$ schedules are distributedly stored at the requesting
nodes: each node knows for each of its requests the node which
invoked the subsequent request in the schedule so that a server after
serving one request can subsequently move to the next node (not
necessarily a different node). As long as new requests are invoked the
schedule is extended. Therefore, in response to the appearance of a
new request at a given processor, \alg\ must contact a processor that
invoked a request but yet has no successor request in the global
schedule, to instruct the motion of the corresponding server. This will
result in the entry of a server to the requesting processor. Sending a
server from a processor to another one is done using an underlying
routing scheme that routes most efficiently. The goal is to minimize
the ratio between the communication costs of an online and an
\textit{optimal offline} protocols that solve \dsms.
We assume that an optimal offline \dsms\ protocol \opt\ knows the whole sequence of requests in advance. However, \opt\ still
needs to send messages from each request to its predecessor request.
The \dsms\ problem has some interesting applications. We state two of them:

\para{Distributed $k$-server problem:} The $k$-server problem
\cite{bansal2011polylogarithmic,manasse1988competitive}, is arguably
one of the most influential research problems in the area of online
algorithms and competitive analysis. The distributed $k$-server was
studied in \cite{bartal1992distributed} where requests arrive
sequentially one by one, but only after the current request is served.
The cost function for this problem is defined as the sum of all
communication costs and the total movement costs of all servers. A
generalization of the $k$-server problem where requests can arrive
over time is called the online service with delay (OSD) problem
\cite{azar2017online,bienkowski2018online}. The OSD cost function is
defined as the sum of the total movement costs of all servers and the
total \textit{delay cost}. The delay of a request is the difference
between the service and the arrival times.

\para{Distributed queuing problem:} This problem is an application of
\dsms\ with $k=1$ , i.e., only one server or shared object
\cite{demmer1998arrow,herlihy2006dynamic,herlihy2001competitive}. The
distributed queuing problem is at the core of many distributed
problems that schedule concurrent access requests to a shared object.
The goal is to minimize the sum of the total communication cost and
the total ``waiting time''. The waiting time of a request is the
difference between the times when the request message reaches the
processor of the predecessor request and when the predecessor request
is invoked. Note that in this problem, the processor of a request must
only send one message to the processor of the predecessor request in
the global schedule. Two well-known applications for this problem are
distributed mutual exclusion
\cite{naimi1987improvement,raymond1989tree,van1987fair} and
distributed transactional memory \cite{zhang2010dynamic}. 

Next, we explain why \dsms\ is also interesting from a theoretical
point of view even for one-shot executions, that is, when all requests
are simultaneously invoked. \Cref{fig:complication} shows a rooted tree $T$, where the lengths of
all edges of each level are equal. Further, the length of every edge
is shorter than the length of its parent edge by some factor larger
than one. A set of six requests arrive at the leaves of $T$ at the
same time. Two servers $s_0$, $s_1$ are initially located at the
points that invoked requests $r_0^1$ and $r_0^2$. Serving the requests
$r_0^1$ and $r_0^2$ does not require communication, and these two
requests are the current tails of the queues of $s_0$ and $s_1$. The
requests $r_0^1$ and $r_0^2$ are at the heads of the two queues. An
optimal solution for serving the remaining requests is that $s_0$
consecutively serves the requests $r_b$, $r_c$, and $r_a$ after
serving $r_0^1$, while $s_1$ serves $r_d$ after having served $r_0^2$.
Next, consider an asynchronous network where message latencies are
arbitrary and protocols have no control over these latencies. A
possible schedule, in this case, is shown in \Cref{fig:complication}:
Request $r_a$ is scheduled after $r_0^1$, $r_b$ after $r_a$, and $r_d$
after $r_b$, since the message latency of a request further away can
be much less than the latency of a closer request. This can lead to
complications with regard to improving the \textit{locality} as it is
met in the above optimal solution.

\begin{figure}[H]
  \center	
  \includegraphics[width=0.3\textwidth]{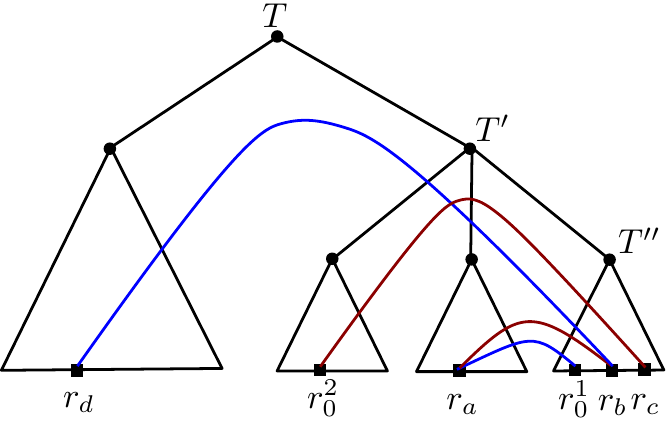}
  \caption{A distributed protocol may lead to complications with regard to improving locality.}
  \label{fig:complication}
\end{figure}

\para{GNN protocol:}
We devise the generalized nearest-neighbor (\gnn) protocol that greedily solves the
\dsms\ problem on \textbf{overlay trees}. An overlay tree $T$ is a
rooted tree that is constructed on top of the underlying network. The
processors of the original network are in a one-to-one correspondence
with the leaves of $T$. Hence, only $T$\textquotesingle s leaves can invoke requests,
and the remaining overlay nodes are artificial. The $k$ servers reside
at different leaves of $T$. Initially, all edges of $T$ are oriented
such that from each leaf there is a directed path to a leaf, where a
server resides. This also implies that every leaf node with a server
has a self-loop. Roughly speaking, the main idea of \gnn\ is to
update the directions of edges with respect to future addresses of a
server. A leaf invoking a request forwards a message along the
directed links, the orientations of all these links are inverted. When
a message reaches a node and finds several outgoing (upward/downward)
links, it is forwarded via an arbitrary downward link to find the
current or a future address of a server. We show that in \gnn\, a
processor holding a request always sends a message through a direct
path to the processor of the predecessor request in the global
schedule. We refer to \Cref{sec:gnn} for a formal description of \gnn.

\subsection{Our Contribution}
\label{sec:contribution}

This paper introduces the \dsms\ problem as a distributed online
allocation problem. We devise the greedy protocol \gnn\ that solves
the \dsms\ problem on overlay trees. We prove that even in an
asynchronous system \gnn\ operates correctly, that is, it does not
suffer from starvation, nor livelocks, or deadlocks. To the best of
our knowledge, \gnn\ is the first link-reversal-based protocol that
supports navigating more than one server.

\begin{theorem}\label{th:correctness}
  Suppose the overlay tree $T$ is constructed on top of a distributed
  network. Consider the \dsms\ problem on $T$ where a set of $k \geq
  1$ identical mobile servers are initially located at different
  leaves of $T$. Further, a sequence of requests can be invoked at any
  time by the leaves of $T$. Then \gnn\ schedules all requests to be
  served by some server at the requested points in a finite time
  despite asynchrony.
\end{theorem}

While \gnn\ itself solves any instance of the \dsms\ problem, we
analyze \gnn\ for the particular case that the requests are
simultaneously invoked. We consider general distributed networks with
$n$ processors. We model such a network by a graph $G$. A
hierarchically well-separated tree (HST) is an overlay tree with
parameter $\alpha >1$, that is, an $\alpha$-HST is a rooted tree where
every edge weight is shorter by a factor of $\alpha$ from its parent
edge weight. A tree is an HST if it is an $\alpha$-HST for some
$\alpha>1$. There is a randomized embedding of any graph into a
distribution over HSTs \cite{bartal96,fakcharoenphol2003tight}. We
sample an HST $T$ according to the distribution defined by the
embedding. We consider an instance $I$ of the \dsms\ problem where the
communication is asynchronous, and the requests are simultaneously
invoked by the nodes of $G$. When running \gnn\ on $T$, we get a
randomized distributed protocol on $G$ that solves $I$ with an
expected competitive ratio of $\bigO{\log n}$ against oblivious
adversaries\footnote{This assumes that the sequence of
  requests is statistically independent of the randomness used for
  constructing the given tree.}.

\begin{theorem}\label{th:graph}
  Let $I$ denote an instance of the \dsms\ problem consisting of an
  asynchronous network with $n$ processors and a set of requests that
  are simultaneously invoked by processors of the network. There is a
  randomized distributed protocol that solves $I$ with an expected
  competitive ratio of $O(\log n)$ against an oblivious adversary.
\end{theorem}

Consider an instance $I$ of \dsms\ that consists of an HST $T$
where communication is asynchronous and a set of requests that are
simultaneously invoked by the leaves of $T$. Analyzing \gnn\ for $I$
turns out to be involved and non-trivial. The fact that the \gnn\ (as
any other protocol) has no control on the message latencies bears a
superficial resemblance to the case where the requests are invoked
over time. Hence, when analyzing \gnn\ for $I$, one faces the
following complications: 1) A server may go back to a subtree of $T$
after having left it. 2) A request in a subtree of $T$ that initially
hosts at least one server can be served by a server that is initially
outside this subtree. 3) Different servers can serve two requests in a
subtree of $T$ that does not initially host any server.
\Cref{th:graph} is derived from our main technical result for HSTs.

\begin{theorem}\label{th:hstOptimal}
 Consider an instance $I$ of \dsms\ that consists of an HST $T$ where even the communication is asynchronous
 and a set of requests that are simultaneously invoked by the leaves of $T$.
 The \gnn\ protocol optimally solves $I$.
\end{theorem}  

One-shot executions of the distributed queuing problem for synchronous
communication were already considered in
\cite{herlihy2001competitive}. The following corollary follows from
\Cref{th:hstOptimal}.

\begin{corollary}\label{co:hstOptimalQueuing}
  \gnn\ optimally solves the distributed queuing problem on HSTs for
  one-shot executions even when the communication is asynchronous.
\end{corollary} 

We provide a simple reduction form the
distributed $k$-server problem to the \dsms\ problem. Our following
lower bound is obtained using this reduction and an existing lower
bound \cite{bartal1992distributed} on the competitive ratio for the
distributed $k$-server problem.

\begin{theorem}\label{thm:LB}
  There is a network topology with $n$ processors---for all $n$---such
  that there is no online distributed protocol that solves \dsms\ with
  a competitive ratio of $o(\max\{k,\log n/\log \log n\})$ against adaptive
  online adversaries where $k$
  is the number of servers. This result even holds when requests are
  invoked one by one by processors in a sequential manner and even
  when the communication is synchronous.
\end{theorem}

\subsection{Further Related Work} 
\label{sec:relatedwork}

\para{Distributed $k$-server problem:} In \Cref{sec:intro}, we have seen that the distributed $k$-server problem
is an application of the \dsms\ problem. In \cite{bartal1992distributed}, a general translator that transforms any deterministic
global-control competitive $k$-server algorithm into a distributed competitive one is provided.
This yields poly$(k)$-competitive distributed protocols for the line, trees, and the ring synchronous network topologies.
In \cite{bartal1992distributed}, a lower bound of $\Omega(\max\{k,(1/D)\cdot(\log n /\log \log n)\})$ on the competitive ratio for the distributed
$k$-server problem against adaptive online adversaries is also provided where $n$ is the number of processors. $D$
is the ratio between the cost
to move a server and the cost to transmit a message over the same distance in synchronous networks.
\cite{azar2017online} and \cite{bienkowski2018online} study OSD on HSTs and lines, respectively.
\cite{azar2017online} provides an upper bound of $\bigO{\log^3n}$ and \cite{bienkowski2018online} provides
an upper bound of $\bigO{\log n}$ on the competitive ratio for OSD where $n$ is the number of leaves of the input
HST as well as the number of nodes of the input line.

\para{Distributed queuing problem and link-reversal-based protocols:}
A well-known class of protocols has been devised based on link reversals to solve distributed problems
in which the distributed queuing problem is at the core of them
\cite{attiya2010provably,khanchandani2019arvy,naimi1987improvement,raymond1989tree,van1987fair,welch2011link,zhang2010dynamic}. 
In a distributed link-reversal-based protocol nodes keep a link pointing to neighbors in the current or future direction of the
server. When sending a message
over an edge to request the server, the direction of the link flips.
We devise the \gnn\ protocol that is---to the best of our knowledge---the first link-reversal-based protocol that navigates more than one server. 
A well-studied link-reversal-based protocol is called \arrow\ \cite{naimi1987improvement,raymond1989tree,van1987fair}.
Several other tree-based distributed queueing protocols that are similar to
\arrow\ have also been proposed. They operate on fixed
trees. The \textsc{Relay} protocol has been introduced as a
distributed transactional memory protocol \cite{zhang2010dynamic}. It is run on top of a fixed
spanning tree similar to \arrow; however, to more efficiently deal with
aborted transactions, it does not always move the shared object to the
node requesting it. Further, in \cite{attiya2010provably}, a
distributed directory protocol called \textsc{Combine} has been
proposed. \textsc{Combine} like \gnn\ runs on a fixed overlay tree, and it is in particular
shown in \cite{attiya2010provably} that \textsc{Combine} is starvation-free.

The first paper to study the competitive ratio of concurrent
executions of a distributed queueing protocol is
\cite{herlihy2001competitive}. It shows that in synchronous
executions of \arrow\ on a tree $T$ for one-shot executions, the total cost of \arrow\ is
within a factor $\bigO{\log m}$ compared to the optimal queueing cost
on where $m$ is the number of requests.
This analysis has later been extended to the general concurrent setting where requests are
invoked over time. In \cite{herlihy2006dynamic},
it is shown that in this case, the total cost of \arrow\ is within a
factor $O(\log D)$ of the optimal cost on $T$ where $D$ is the diameter of
$T$. Later, the same bounds have also been proven for \textsc{Relay}
\cite{zhang2010dynamic}. Typically, these protocols are run on a
spanning tree or an overlay tree on top of an underlying general
network topology. In this case, the competitive ratio becomes
$O(s\cdot \log D)$, where $s$ is the stretch of the tree. Finally, \cite{ghodselahi2017dynamic} has shown
that when running \arrow\ on top of HSTs, a randomized distributed online queueing protocol
is obtained with expected competitive ratio $O(\log n)$ against an oblivious adversary even on general
$n$-node network topologies. The result holds even if the queueing
requests are invoked over time and even if communication is asynchronous.
The main technical result of the paper shows that the competitive ratio of \arrow\ is
constant on HSTs.

\para{Online tracking of mobile users:} A similar problem to \dsms\ is the online mobile user tracking problem \cite{awerbuch1995}.
In contrast with \dsms\ where a request $r$ results in moving a server to the requesting point,
here the request $r$ can have two types: find request that does not result in moving the mobile user and move request.
A request in \dsms\ that is invoked by $v$ can be seen as a combination of a find request that is invoked at $v$
in the mobile user problem and a move request invoked at the current address of the mobile user. 
The goal is to minimize the sum of the total communication cost and the total cost incurred for moving the mobile user. 
\cite{awerbuch1995} provides an upper bound of $\bigO{\log^2 n}$ on the competitive ratio for the online
mobile user problem for one-shot executions. Further, \cite{alon1992lower} provides a lower bound of $\Omega(\log n/\log \log n)$
on the competitive ratio for this problem against an oblivious adversary. 
\section{Model, Problem Statement, and Preliminaries}
\label{sec:model}

\subsection{Communication Model}
\label{sec:communicationModel}
We consider a point-to-point communication network that is modeled by
a graph $G=(V,E)$, where the $n$ nodes in $V$ represent the
processors of the network and the edges in $E$ represent bidirectional
communication links between the corresponding processors. We suppose
that the edge weights are positive and are normalized such that the
weight of each edge will be at least $1$. If $G$ is unweighted, then we
assume that the weight of an edge is $1$. We consider the message
passing model \cite{peleg2000distributed} where neighboring processors
can exchange messages with each other.
The communication links can have different latencies. These latencies
are not even under control of an optimal offline distributed protocol.
We consider both synchronous and asynchronous systems. In a
synchronous system, the latency for sending a message over an edge
equals the weight of the edge. In an asynchronous system, in contrast,
the messages arrive at their destinations after a finite but unbounded
amount of time. Messages that take a longer path may arrive earlier,
and the receiver of a message can never distinguish whether a message
is still in transit or whether it has been sent at all. For our
analysis, however, we adhere to the conventional approach where the
latencies are scaled such that the latency for sending a message over
an edge is upper bounded by the edge weight in the ``worst case'' (for
every legal input and in every execution scenario) (see Section 2.2 in
\cite{peleg2000distributed} for more information).   

\subsection{Distributed Serving with Mobile Servers (\dsms) Problem}
\label{sec:dsms}
The input for \dsms\ problem for a graph $G$ consists of $k\geq 1$
identical mobile servers that are initially located at different nodes
of $G$ and a set $\R$ of requests that are invoked at the nodes at any
time\label{no:requests}. A request $r_i \in \R$ is represented by
$(v_i,t_i)$ where node $v_i$ invoked request $r_i$ at time $t_i \geq
0$\label{no:singleRequest}. A distributed protocol \alg\ that solves
the \dsms\ problem needs to serve each request with one of the $k$
servers at the requested node. Hence, \alg\ must schedule all requests
that access a particular server. Consequently, \alg\ outputs $k$
global schedules such that the request sets of these schedules form a
partition of $\R$ and all requests of the schedule $\pi^z_{\alg}$
consecutively access the server $s^z$ where $z \in
\set{1,\dots,k}$\label{no:schedule}\label{no:server}. We assume that
at time $0$, when an execution starts, the tail of schedule
$\pi^z_{\alg}$ is at a given node $v^z_0\in V$ that hosts $s^z$.
Formally, this is modeled as a ``dummy request'' $r^z_0=(v^z_0,0)$
that has to be scheduled first in the schedule $\pi^z_{\alg}$ by
\alg\label{no:dummy}. Consider two requests $r_i$ and $r_j$ that are
consecutively served by $s^z$ where $r_i$ is scheduled after $r_j$. To
schedule request $r_i$ the protocol needs to inform node $v_j$, the
predecessor request $r_j$ in the constructed schedule.
As soon as $r_j$ is served by $s^z$, node $v_j$ sends the server to
$v_i$ for serving $r_i$ using an underlying routing facility that
efficiently routes messages. The goal is to minimize the total
communication cost, i.e., the sum of the latencies of all messages
sent during the execution of \alg.

\subsection{Preliminaries}
\label{sec:preliminaries}
Consider a distributed protocol \alg\ for the \dsms\ problem when
requests can arrive at any time. Let $\R$ denote the set of requests,
including the dummy requests. Assume that \alg\ partitions $\R$ into
$k$ sets $\R^1_{\alg},\dots,\R^k_{\alg}$, and that it schedules the
requests in set $\R^z_{\alg}$ according to permutation $\pi^z_{\alg}$.
Denote the request at position $i$ of $\pi^z_{\alg}$ by
$r_{\pi^z_{\alg}(i)}$\label{no:requestPartition}. The dummy request
$r^z_0$ of $\pi^z_{\alg}$ is represented by $r_{\pi^z_{\alg}(0)}$. Let
$\latency_{\alg}(\mu)$ denote the latency of message $\mu$ as routed
by \alg\label{no:latency}. For every $i \in \set{1,\dots,|\R|-1}$, if
$r_i$ belongs to $\R^z_{\alg}$, the communication cost
$\delay_{\alg}\big(r_{\pi^z_{\alg}(i-1)},r_{\pi^z_{\alg}(i)}\big)$\label{eq:delayCost}
incurred for scheduling $r_{\pi^z_{\alg}(i)}$ as the successor of
$r_{\pi^z_{\alg}(i-1)}$ is the sum of the latencies of all messages sent
by \alg\ to schedule $r_{\pi^z_{\alg}(i)}$ immediately
after $r_{\pi^z_{\alg}(i-1)}$. The total communication cost of \alg\
for scheduling all requests in $\R^z_{\alg}$ is defined as
\begin{equation}
  \Delay_{\alg}(\pi^z_{\alg}) := 
         \sum_{i=1}^{|\R^z_{\alg}|-1}\delay_{\alg}\left(r_{\pi^z_{\alg}(i-1)},r_{\pi^z_{\alg}(i)}\right).\label{eq:componentTotalCost}
\end{equation}
The total communication cost of \alg \ for scheduling all requests in $\R$, therefore, is
\begin{equation}\label{eq:totalCost}
	\Delay_{\alg} := \sum_{z=1}^{k} \Delay_{\alg}(\pi^z_{\alg}).
\end{equation}

\subsection{Hierarchically Well-Separated Trees (HSTs)}
\label{sec:hst}

Embedding of a metric space into probability distributions over tree
metrics have found many important applications in both centralized and
distributed settings
\cite{azar2017online,bansal2011polylogarithmic,ghodselahi2017dynamic}.
The notion of a hierarchically well-separated tree was defined by Bartal in \cite{bartal96}.

\begin{definition}[\textbf{$\alpha$-HST}]\label{de:hst}
  For $\alpha>1$ an $\alpha$-HST of depth $h$ is a rooted tree with
  the following properties: The children of the root are at a distance
  $\alpha^{h-1}$ from the root and every subtree of the root is an
  $\alpha$-HST of depth $h-1$. A tree is an HST if it is an
  $\alpha$-HST for some $\alpha>1$.
\end{definition}

The definition implies that the nodes two hops away from the root are
at a distance $\alpha^{h-2}$ from their parents. The probabilistic
tree embedding result of \cite{fakcharoenphol2003tight} shows that for
every metric space $(X,d)$ with minimum distance normalized to $1$ and
for every constant $\alpha> 1$ there is a randomized construction of
an $\alpha$-HST $T$ with a bijection $f$ between the points in $X$ and
the leaves of $T$ such that a) the distances on $T$ are dominating the
distances in the metric space $(X,d)$, i.e., $\forall x,y\in X:
d_T\big(f(x),f(y)\big)\geq d(x,y)$ and such that b) the expected tree
distance is $\E\big[d_T\big(f(x),f(y)\big)\big] = O(\alpha \log
|X|/\log \alpha)\cdot d(x,y)$ for every $x,y\in X$. The length of the
shortest path between any two leaves $u$ and $v$ of $T$ is denoted by
$d_T(u,v)$\label{no:shortestDist}. An efficient distributed
construction of the probabilistic tree embedding of
\cite{fakcharoenphol2003tight} has been given in
\cite{ghaffari2014near}.
\section{The Distributed GNN Protocol}
\label{sec:gnn}
In this section the \gnn\ protocol is introduced. 
\subsection{Description of GNN}
\label{sec:gnnDescription}

\gnn\ runs on overlay trees and outputs a feasible
solution for the \dsms\ problem. Consider a rooted tree $T=(V_T,E_T)$
whose leaves correspond to the nodes of the underlying graph
$G=(V,E)$, i.e., $V \subseteq V_T$. Let $n=|V|$. The $k \geq 1$
identical mobile servers are initially at different leaves of $T$.
Further, there is a dummy request at every leaf that initially hosts a
server. The leaves of $T$ can invoke requests at any time. A leaf node
can invoke a request while it is hosting a server and a leaf can also
invoke a request while its previous requests have not been served yet.
Initially, a directed version of $T$ is constructed and denoted by
$H$\label{no:H}, the directed edges of $H$ are called {\em links}.
During an execution of \gnn, \gnn\ changes the directions of the links. Denote by
$v.\pointer$ the set of neighbors of $v$ that are pointed by $v$. After a leaf $u$
has invoked a request it sends a find-predecessor message denoted by
$\mu(u)$ along the links to inform the node of the predecessor
request in the global schedule\label{no:messageSender}. The routing of
$\mu(u)$ is explained below.  At the beginning before any message is sent and for any server, all the nodes
on the direct path from the root of $T$ to the leaf that hosts the server, point to the server. Further, the
host points to itself and creates a self-loop. Hence, we have $k$ directed paths with downward links from the
root of $T$ to the points of the current tails of the schedules. Any other node points to its parent with an
upward link. Therefore, the sets $v.\pointer$ for all $v \in V_T$ are non-empty at the beginning of the executing the protocol. \Cref{fig:initial_gnn} shows the directed HST at the beginning as an example.

\begin{algorithm}[H]
\DontPrintSemicolon
\SetKwInOut{Input}{Input}
\Input{The rooted tree $T$, $k$ identical mobile servers that are initially
at distinct leaves of $T$, and a set of requests that are invoked over time}

\BlankLine

\SetKwInOut{Output}{Output}
\SetAlgoNoLine
\Output{$k$ schedules for serving all requests}

\BlankLine

\nonl \textbf{Upon requesting a service: Algorithm \ref{alg:newReq}}\;

\BlankLine

\nonl \textbf{Upon receiving a find-predecessor message: Algorithm \ref{alg:findReceived}}\;

\caption{\gnn\ Protocol}
\label{alg:gnn}
\end{algorithm}

\para{Upon $u$ invoking a new request:} Consider the leaf node $u$ when it invokes a new request $r$.
If $u$ has a self-loop, then $r$ is scheduled immediately behind the last request that has been invoked at $u$.
Otherwise, the leaf $u$ atomically sends $\mu(u)$ to its parent through an upward
link, $u$ points to itself, and the link from $u$ to its parent is removed.
We suppose that messages are reliably delivered. The details of this part of the protocol are given by
\Cref{alg:newReq}. See \Cref{fig:initiate_gnn} as an example.

\begin{algorithm}[H]
\DontPrintSemicolon
\SetAlgoNoLine

\textbf{do atomically}\;  
\tcc{suppose $u.\pointer=\set{v}$ ($u$ as a leaf always points either to itself or to its parent)}
\pushline \eIf{$u=v$}
{    
	$r$ is scheduled immediately after the last request that has been invoked by $u$\;
}
{  
	$u$ sends $\mu(u)$ to $v$\;
	$u.\pointer:=\set{u}$\; \label{le:ClaimUniqueEdge1}   
}
\popline \textbf{end}\;

\caption{Upon $u$ invoking a new request $r$}
\label{alg:newReq}
\end{algorithm}

\begin{figure}[H]
    \centering
    \begin{subfigure}{0.3\textwidth}
      \centering
      \includegraphics[width=\textwidth]{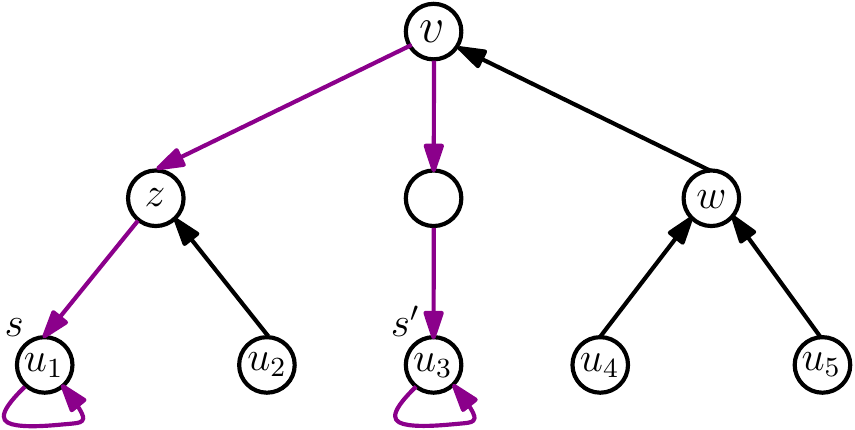}
      \caption{initial system state}
      \label{fig:initial_gnn}
    \end{subfigure}
    \hfill
    \begin{subfigure}{0.3\textwidth}
      \centering
      \includegraphics[width=\textwidth]{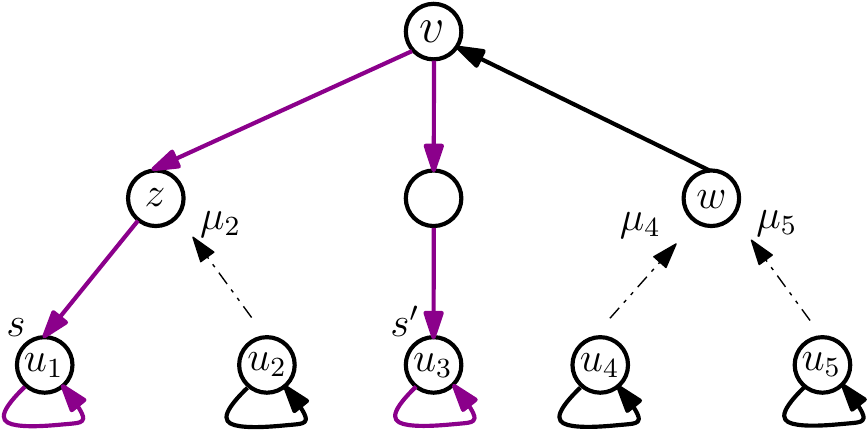}
      \caption{step 1}
      \label{fig:initiate_gnn}
      \end{subfigure}
    \hfill
    \begin{subfigure}{0.3\textwidth}
      \centering
      \includegraphics[width=\textwidth]{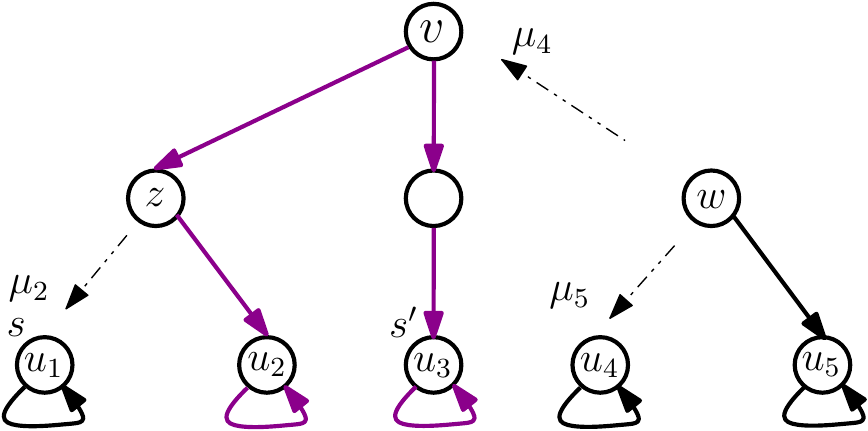}
      \caption{step 2}
      \label{fig:intermediate1_gnn}      
      \end{subfigure}     
      \caption{\gnn\ protocol: (a) The servers $s$ and $s'$
      serve requests in schedules $\pi$ and $\pi'$, respectively. The dummy requests at $u_1$
      and $u_3$ are the initial tails of $\pi$ and $\pi'$, respectively. (b) Nodes $u_2$, $u_3$, $u_4$, and $u_5$ respectively issue requests
      $r_2$, $r_3$, $r_4$, and $r_5$ at the same time and send the find-predecessor messages $\mu_2$, $\mu_3$, $\mu_4$, and $\mu_5$, respectively,
      along the arrows. (c) The request $r_3$ is the current tail of $\pi'$.
    Both $\mu_4$ and $\mu_5$ reach  $w$ at the same time. First, the message $\mu_4$ is arbitrarily processed by $w$
    and $w$ forwards $\mu_4$ towards $v$ and therefore $\mu_5$ is deflected towards $u_4$.}
    \label{fig:intermediate001_gnn}
\end{figure}

\para{Upon $w$ receiving $\mu(u)$ from node $v$:}
Suppose that node $w$ receives a find-predecessor message $\mu(u)$ from node $v$.
The node $w$ executes the following steps atomically. If $w$ has at least one downward link,
then $\mu(u)$ is forwarded to some child of $w$ through a downward link (ties are broken
arbitrarily).
Then, $w$ removes the
downward link and adds a link to $v$---independently of whether $v$ is the parent or a child of $w$.
If $w$ does not have a downward link, it either points to itself, or it has an upward link.
In the latter case,
$\mu(u)$ is atomically forwarded to the parent of $w$, the upward link from $w$ to its parent is removed and
then $w$ points to $v$ using a downward link.
Otherwise, $w$ is a leaf and points to itself. The request $r$ invoked
by $u$ is scheduled behind the last request that has been invoked by $w$. Then, $w$ removes the link that points to itself
and points to $v$ using an upward link. The details of this part of the protocol are given by
\Cref{alg:findReceived}. Also, see \Cref{fig:intermediate1_gnn} and \Cref{fig:intermediate234_gnn}.

\begin{algorithm}[H]
\DontPrintSemicolon
\SetAlgoNoLine
\textbf{do atomically}\; 
  
\pushline \eIf{there exists a child node in $w.\pointer$}
{ 
$z=:$ an arbitrary child node in $w.\pointer$\;
}
{
$z=:$ the only node in $w.\pointer$\;
}
$w.\pointer:=w.\pointer - \set{z}$\; \label{le:startClaimUniqueEdge2}
$w.\pointer:=w.\pointer \cup \set{v}$\; \label{le:endClaimUniqueEdge2}
\eIf{$z \neq w$}
{
	\pushline $w$ sends $\mu(u)$ to $z$\;
}
{
	the corresponding request to $\mu(u)$ is scheduled immediately after the last request that has been invoked by $w$\;
}
\popline \textbf{end}\; 

\caption{Upon $w$ receiving $\mu(u)$ from node $v$ ($w \neq v$)}
\label{alg:findReceived}
\end{algorithm} 

\begin{figure}[H]
    \centering
    \begin{subfigure}{0.3\textwidth}
      \centering
      \includegraphics[width=\textwidth]{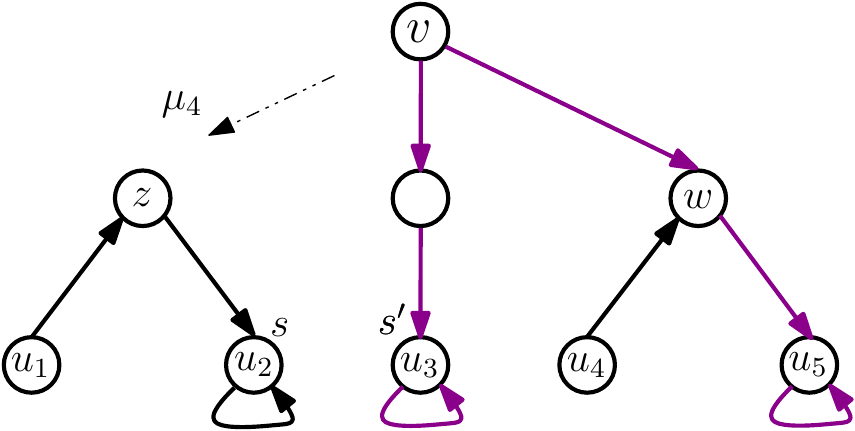}
      \caption{step 3}
      \label{fig:intermediate2_gnn}
    \end{subfigure}
    \hfill
    \begin{subfigure}{0.3\textwidth}
      \centering
      \includegraphics[width=\textwidth]{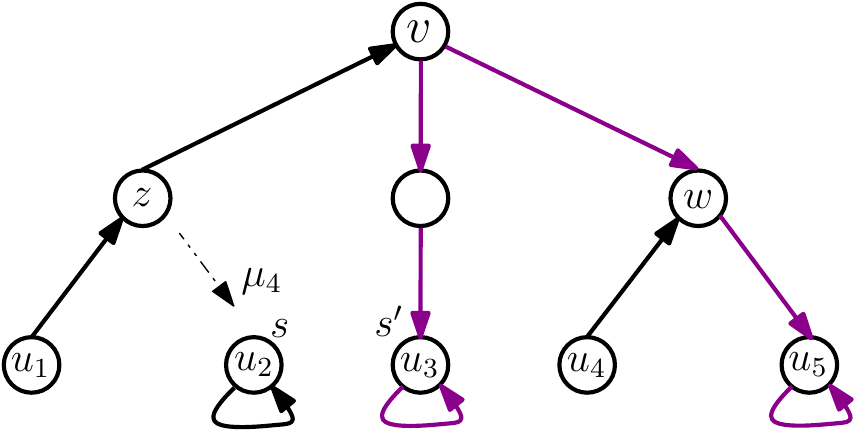}
      \caption{step 4}
      \label{fig:intermediate3_gnn}
    \end{subfigure}
    \hfill
    \begin{subfigure}{0.3\textwidth}
      \centering
      \includegraphics[width=\textwidth]{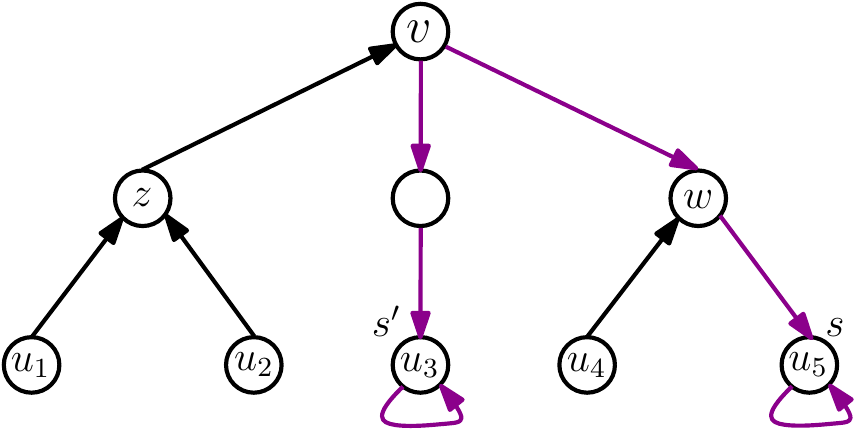}
      \caption{step 5}
      \label{fig:intermediate4_gnn}
    \end{subfigure}
    \caption{\gnn\ protocol: (a) The request $r_2$ is scheduled behind
    the current tail of $\pi$ and now $r_2$ is the current tail of $\pi$ and $u_2$ obtains the server $s$.
	The request $r_5$ is scheduled behind $r_4$ while $\mu_4$ is still in transit.
	(b) $\mu_4$ still follows arrows, reversing the directions of arrows along its way.
	(c) The request $r_4$ is scheduled behind $r_2$ and $s$ moves to
      $u_4$. After $r_4$ served by $s$, then $s$ moves from $u_4$ to $u_5$
      since $r_5$ has already been scheduled behind $r_4$. \Cref{fig:intermediate001_gnn}--\Cref{fig:intermediate234_gnn}
      illustrates that there is always at least one connected path with purple
      arrows from the root to some leaf.}
    \label{fig:intermediate234_gnn}
\end{figure}

\subsection{Correctness of GNN}
\label{sec:gnnCorrectness}

Regarding the description of \gnn, we need to show two invariants for \gnn. The first is that \gnn\
eventually schedules all requests. The second one is that \gnn\ is starvation-free so that
a scheduled request is eventually served. 

\subsubsection{Scheduling Guarantee}
\label{sec:scheduling}
\begin{theorem}\label{th:scheduleGuarantee}
\gnn\ guarantees that the find-predecessor message of any node that invokes a request always
reaches a leaf node $v$ in a finite time such that $v.\pointer=\set{v}$. 
\end{theorem}

We prove the scheduling guarantee stated in \Cref{th:scheduleGuarantee} using the following properties of \gnn. First, we need to
show that any node always has at least one outgoing edge in \gnn.

\begin{lemma}\label{le:uniqueEdge}
In \gnn, $v.\pointer$ is never empty for any node $v \in V_T$.
\end{lemma}
\begin{proof}
At the beginning of any execution, $v.\pointer$ is not empty for any $v \in V_T$.
The set $v.\pointer$ changes only when there is a (find-predecessor) message at $v$ (see Line \ref{le:ClaimUniqueEdge1} of \Cref{alg:newReq}
and Line \ref{le:startClaimUniqueEdge2} and Line \ref{le:endClaimUniqueEdge2} of \Cref{alg:findReceived}).
During an execution, every time $v$ receives a message, a node is removed from $v.\pointer$ while a new node is added to
$v.\pointer$. This also covers the case when at least two messages are received by $v$ at the same time. The node $v$ atomically processes all these
messages in an arbitrary order. Therefore, $v.\pointer$ never gets empty.    
\end{proof}

\begin{lemma}\label{le:edgeStates}
\gnn\ always guarantees that on each edge of $H$, there is either exactly one link or exactly one message in transit. 
\end{lemma}  
\begin{proof}
Initially, either a node points to its parent with an upward link or a node points to its children with downward links in the \gnn\ protocol.
Consider the edge $(u,v)$ where $v \in u.\pointer$. Further, consider the first time in which a message is in transit on $(u,v)$. Immediately
before this transition occurs, $u$ must point to $v$, and there is not any message in transit on the edge. Therefore, w.r.t. the protocol description, the
message must be sent by $u$ to $v$, and the link that points from $u$ to $v$ has been removed. Since there is not any link while
the message is in transit, it is not possible to have a second message to be in transit at the same time. When the message arrives at $v$, the node
$v$ points to $u$, and the message is removed from the edge. The next time, if a message will be transited on the edge, then $v$ must have sent it to $u$
and removed the link that points from $v$ to $u$. 
\end{proof}

\begin{lemma}\label{le:acyclic}
The directed tree $H$ always remains acyclic during an execution,
hence a path from a node to another node in $H$ is always the direct path.
\end{lemma}
\begin{proof}
The \gnn\ protocol runs on the directed tree $H$ in which the underlying tree---that is, $T$---is fixed, and the directions of links on $H$ are only changed.
Therefore, $H$ is acyclic because the tree is always fixed, and w.r.t. \Cref{le:edgeStates} that shows that it never occurs a state where on
the edge $(u,v)$, $u$ and $v$ point to each other at the same time. 
\end{proof}

The following lemma implies that a find-predecessor message always
reaches the node of its predecessor using a direct path constructed by \gnn.

\begin{lemma}\label{le:uniquePath}
\gnn\ guarantees that there is always at least one direct path in $H$ from any leaf node $u$ to a
leaf node $v$ where $v.\pointer=\set{v}$.
\end{lemma}
\begin{proof}
If the leaf node $u$ points to itself, we are done. Otherwise, w.r.t. \Cref{le:uniqueEdge} there must be a path from $u$ to a leaf node $v$ since
the tree $H$ is acyclic. This path must be a direct path w.r.t. \Cref{le:acyclic}. The leaf node $v$ must point to itself w.r.t. \Cref{le:uniqueEdge}.
\end{proof}

\begin{proof}[\textbf{Proof of \Cref{th:scheduleGuarantee}}]
Using \Cref{le:uniquePath}, it remains to show that any message traverses a direct path between two leaves in a finite time.
The number of edges on the direct path between any two leaves of $T$ is upper bounded by the diameter of the tree. Further, any
message that is in transit at edge $(u,v)$ from $u$ to $v$ is delivered reliably at $v$ in a finite time. Therefore, to show that a request is
eventually scheduled in a finite time, it remains to show that a message will never be at a node for the second time.
To obtain a contradiction, assume that the message $\mu$ is the first message that visits a node twice, and the first node visited
twice by $\mu$ denoted by $v \in V_T$. With respect to \Cref{le:acyclic}, there is never a cycle in $H$. Therefore, the edge $e=(u,v)$ must
be the first edge that is traversed
by $\mu$ first from $v$ to $u$ and immediately from $u$ to $v$ for the second time, and $\mu$ must be the first message that traverses an edge twice.
This implies that immediately before $u$ receives $\mu$, the node $u$ points to $v$, and $\mu$ is in transit on $e$ at the same time. This is a contradiction
with \Cref{le:edgeStates}.    
\end{proof}

\subsubsection{Serving Guarantee}
\label{sec:serving}
\begin{theorem}\label{th:servingGuarantee}
\gnn\ is starvation-free. In other words, any scheduled request is eventually served by some server. 
\end{theorem}
Consider any of $k$ global schedules that produced by \gnn, say $\pi^w_{\gnn}$. 
Assume that there is more than one request scheduled in $\pi^w_{\gnn}$. For any two requests $r_i=(v_i,t_i)$ and $r_j=(v_j,t_j)$
in $\pi^w_{\gnn}$ where $r_i$ is scheduled immediately before $r_j$, we see $e=(r_i,r_j)$ as a directed edge
where $r_j$ points to $r_i$. This edge is actually
simulated by the direct path---by \Cref{le:uniquePath}, a message always finds the node of its predecessor
using a direct path on $H$---between the leaves $v_i$ and $v_j$ that is traversed by the message sent from $v_j$ to $v_i$.
Let $F^w_{\alg}$ denote the graph constructed by the messages of all requests in $\R^w_{\alg}$.
 
\begin{lemma}\label{le:queuDirectPath}
$F^w_{\alg}$ is a directed path towards the head of the schedule, that is, $r^w_0=r_{\pi^w_{\gnn}(0)}$. 
\end{lemma}
\begin{proof}
The proof has three parts.
\begin{enumerate}[1)]
\item \textbf{Any node of $F^w_{\alg}$, except the dummy request, has exactly one outgoing edge:} This is obvious since any node that invokes a request
sends exactly one message. 

\item \textbf{Any node in $F^w_{\alg}$ has at most one incoming edge:} For the sake of
contradiction, assume that there is a node contained in $F^w_{\alg}$ denoted by $r=(v,t)$ with at least two incoming edges in $F^w_{\alg}$. This implies that two
messages must reach $v$ in $H$ before $v$ invokes any other request after $r$. However, when the first message
reaches $v$---if any other message does not reach $v$ before these two messages---$v$ removes the link that points to
itself and adds a link that points to its parent w.r.t. Line \ref{le:startClaimUniqueEdge2} and Line \ref{le:endClaimUniqueEdge2} of \Cref{alg:findReceived}.
The second message cannot reach $v$ as long as at least one request is invoked by $v$ after invoking $r$. This contradicts our assumption in which
two messages reach $v$ before the time when $v$ invokes another request after invoking $r$.

\item \textbf{$F^w_{\alg}$ is connected:}
To obtain a contradiction, assume that the graph $F^w_{\alg}$ is not connected.
Hence, w.r.t. the first and second parts, we have at least one connected component with at least two requests in $\R^w_{\alg}$ that form a cycle,
and the connected component does not include the dummy request in $r^w_0$.
Let $\R^{w,c}_{\alg}$ denote the requests in the connected component $F^{w,c}_{\alg}$ that forms a cycle.
Consider the node $z$ in $V_T$ that is the lowest common ancestor of those leaves of $H$ that invoke the requests in $\R^{w,c}_{\alg}$.
Further, let the subtree $H^{w,c}$ of $H$ denote the tree rooted at $z$.
All messages of requests in $\R^{w,c}_{\alg}$ must traverse inside $H^{w,c}$ since $F^{w,c}_{\alg}$ is disconnected with any request in
$\R^w_{\alg} \setminus \R^{w,c}_{\alg}$.

Assume that at least one message of requests in $\R^{w,c}_{\alg}$ reaches $z$.
Consider the first message $\mu$ by $r$ that reaches $z$ at time $t$.
If there is not any downward link at $z$ at $t$, then $\mu$ is forwarded to the parent of $z$.
This is a contradiction with the fact that $F^{w,c}_{\alg}$ is disconnected with any request in $\R^w_{\alg} \setminus \R^{w,c}_{\alg}$.
Hence, there must be at least one downward link at $z$ at $t$.
On the other hand, since $\mu$ is the first message of requests in $\R^{w,c}_{\alg}$ that reaches $z$,
all downward links at $z$ at time $t$ must have been created by some messages of requests in $H^{w,c}$ that are not in $\R^{w,c}_{\alg}$.
Note that if a downward link at $z$ is there since the beginning, then we assume that, w.l.o.g.,
it has been created by a ``virtual message'' sent by the node of the corresponding dummy request.
Suppose $\mu$ is forwarded through one of these downward links that was created by the message of
$r'$---as mentioned, $r'$ can be a dummy request---that is in $H^{w,c}$ but not in $\R^{w,c}_{\alg}$.
The original downward path from $z$ to the leaf node of $r'$ can be changed by the message of a request in
$H^{w,c}$---can be a request in $\R^{w,c}_{\alg}$.   
Thus, either $r$ is scheduled immediately behind some request in $H^{w,c}$ that is not in $\R^{w,c}_{\alg}$
or some other request in $\R^{w,c}_{\alg}$.
In either case, we get a contradiction with our assumption in which $F^{w,c}_{\alg}$ is disconnected with any request in $\R^w_{\alg} \setminus \R^{w,c}_{\alg}$. 

If there is not any message of a request in $\R^{w,c}_{\alg}$ that can reach $z$, then there must be at least two downward links during the execution at $z$
that have been created by some messages of requests that are not in $\R^{w,c}_{\alg}$---this holds because if there is at most one downward link
at $z$, then a message of some request in $\R^{w,c}_{\alg}$ must reach $z$ w.r.t. the definition of $z$.
However, the existence of at least two downward links at $z$ implies that $F^{w,c}_{\alg}$ is not connected.
This is true because there are at least two downward paths that partition the requests in $\R^{w,c}_{\alg}$
into two disjoint components in $F^w_{\alg}$ w.r.t. the definition of $z$ and our assumption in which there is not
any message of request in $\R^{w,c}_{\alg}$ that can reach $z$.
This is a contradiction with our assumption in which $F^{w,c}_{\alg}$ is a connected component.
\end{enumerate}
The above three parts all altogether show that $F^w_{\alg}$ is indeed a directed path that points towards the dummy request in $\R^w_{\alg}$.
\end{proof}

\begin{proof}[\textbf{Proof of \Cref{th:servingGuarantee}}]
Consider any of $k$ global schedules that is resulted by \gnn, say $\pi^w_{\gnn}$. If there is only one request in $\pi^w_{\gnn}$---there must
be at least one request, that is the dummy request $r^w_0$---then we are done. Otherwise, w.r.t. \Cref{le:queuDirectPath}
there is a path of directed edges such as $e=(r_i,r_j)$ over the requests in $\R^w_{\alg}$. When $v_i$ obtains a server, and after $r_i$ is served,
$v_i$ sends the server to $v_j$ for serving $r_j$ using an underlying routing scheme. Consequently, all requests in $\R^w_{\alg}$ are served.
\end{proof}

\begin{proof}[\textbf{Proof of \Cref{th:correctness}}]
\Cref{th:scheduleGuarantee} and \Cref{th:servingGuarantee} both together prove the claim of the theorem. 
\end{proof}
\section{Analysis in a Nutshell}
\label{sec:compactanalysis}

From a technical point of view, we achieve our main result on HSTs. In this section,
we provide an analysis of \gnn\ on HSTs in a nutshell.
The complete analysis, including all proofs, appears in \Cref{sec:analysis}.
Our analysis of \gnn\ for general networks appears in \Cref{sec:DSMSonGeneralNetworks}.
The lower bound claimed in \Cref{thm:LB} is proved in \Cref{sec:LB}.

Let \alg\ denote a particular distributed \dsms\ protocol
that sends a unique message from the node of a request to the node of the predecessor request for scheduling the request
(the message can be forwarded by many nodes on the path between the two nodes of the predecessor and successor requests).
Consider a one-shot execution of \alg\ where requests are invoked at the same time $0$.
Let $G=(V,E)$ denote the input graph.
Further, let $B=\big(V_B=\R,E_B={\R \choose 2}\big)$ be the complete graph, and consider two requests $r=(v,0)$ and $r'=(v',0)$ in $\R$ where $v,v' \in V$\label{no:B}.
Assume that $r'$ is scheduled as the successor of $r$ by \alg\ in the global schedule, and w.r.t. the \dsms\ problem definition
\alg\ informs $v$ by sending the (find-predecessor) message $\mu'$ from $v'$ to $v$.
Therefore, the communication cost for scheduling $r'$ equals the latency of $\mu'$. Formally,
\begin{equation}\label{eq:delayOneshot}
	\delay_{\alg}(r,r')=\latency_{\alg}(\mu').
\end{equation}
Let $r_{\src}(\mu')=r'$ denote the request corresponding with $\mu'$. Further, let $r_{\des}(\mu')=r$
denote the predecessor request $r$ in the global schedule\label{no:endpointsMSG}. We see $e=(r,r')$ as an edge in $E_B$ that
is constructed by $\mu'$.
Let us add $\mu(e)$ and $e(\mu)$ to the notation where $\mu(e)$ is the message that constructs the edge $e$ and $e(\mu)$ is the edge
that is constructed by $\mu$\label{no:messageEdge}.
For instance, here, $\mu(e)$ refers to $\mu'$ and $e(\mu')$ refers to the edge $(r,r')$.

\para{Representing solution of ALG as a forest:}
We observe that any of the $k$ resulted schedules $\pi^1_{\alg},\dots,\pi^k_{\alg}$
can be seen as a \textbf{TSP path} that spans all requests in the corresponding schedule as follows (see \Cref{le:queuDirectPath}).
The TSP path $F^z_{\alg}$ starts with the dummy request $r^z_0$ that is the head of $\pi^z_{\alg}$, and a request on the TSP path $F^z_{\alg}$
is connected using an edge to its predecessor in the schedule $\pi^z_{\alg}$.
As mentioned, the edge is constructed by the message sent by the requesting node to the node of its predecessor request.
Therefore, an edge of any TSP path---that is an edge in $E_B$---is actually a path on the input graph that is traversed by
the corresponding message.
For any $F \subseteq \forest_{\alg}$, we define the \textbf{total communication cost of $F$} as follows.
\begin{equation}\label{eq:totalLatencyF}
	\Latency_{\alg}(F) := \sum_{e \in F} \latency_{\alg}\big(\mu(e)\big).
\end{equation}
Therefore, the \textbf{total communication cost of a TSP path} equals the sum of latencies of all messages that construct the TSP path.
The $k$ TSP paths represent a forest of $B$.
Let $\forest_{\alg}$ be the forest that consists of the $k$ TSP paths $F^1_{\alg},F^2_{\alg}, \dots, F^k_{\alg}$ constructed by \alg\label{no:resultedForest}.
We slightly abuse notation and identify a subgraph $F$ of $B=\big(\R,{\R \choose 2}\big)$ with the set of edges contained in $F$.
The \textbf{total communication cost of $\forest_{ALG}$} equals the sum of total costs of the $k$ TSP paths $F^1_{\alg},F^2_{\alg}, \dots, F^k_{\alg}$.
For the input graph $G=(V,E)$, we denote the \textbf{weight of edge} $e=(r,r') \in E_B$ by $w_G(e):=d_G(v,v')$
where $v,v' \in V$\label{no:weightEdge} (recall $r=(v,t)$ and $r'=(v',t')$).
Note that that $d_G(v,v')$ is the weight of the shortest path between $v$ and $v'$ on the input graph $G$.
Generally, \textbf{the total weight of the subgraph} $F$ of $B$ w.r.t. the input graph $G$ equals the sum of weights of all edges in
$F$. Formally,
\begin{equation}\label{eq:forestAlgLength}
	\Weight_G(F) := \sum_{e \in F} \weight_G(e).
\end{equation}
\begin{definition}[\textbf{$S$-Respecting $m$-Forest}]\label{de:SRespectingLForest}
Let $G=(V,E)$ be a graph and $m \leq |V|$. A forest \forest\ of $G$ is called an $m$-forest
if \forest\ consists of $m$ trees. Further, let $S \subseteq V$ , $|S| \leq m$ be a set of at most
$m$ nodes. An $m$-forest \forest\ of $G$ is $S$-respecting if the nodes in
$S$ appear in different trees of \forest.
\end{definition}
Let $\R_D$ denote the set of $k$ dummy requests in $\R$\label{no:RD}. W.r.t. the \Cref{de:SRespectingLForest}, $\forest_{\alg}$ is an
$\R_D$-respecting spanning $k$-forest of $B=\big(\R,{\R \choose 2}\big)$. \textit{From now on, we consider the HST $T$
as the input graph\label{no:hstT}.}

\para{Locality-based forest:}
For any subtree $T'$ of $T$ and any subgraph $F$ of $B$, let $F(T')$
denote the subgraph of $F$ that is induced by those requests contained in $F$ that are also in $T'$\label{no:subgraphTree}. 
Further, let $F^1,F^2, \dots, F^k$ denote the $k$ trees of the spanning $k$-forest \forest\ of $B$.
Consider any $\R_D$-respecting spanning $k$-forest of $B$ with the following basic \textbf{locality properties}\label{no:LBF}.
\begin{enumerate}[I.] 
\item \label{pr:ForestIntracomponent} [\textbf{Intra-Component Property}] For any subtree $T'$ of $T$
and for any $w \in \set{1,\dots,k}$, the component $F^w_{\grd}(T')$ is a tree.
\item \label{pr:ForestIntercomponent} [\textbf{Inter-Component Property}] For any subtree $T'$ of $T$, suppose that
there are at least two non-empty components $F^z_{\grd}(T')$ and $F^w_{\grd}(T')$ where $w \neq z$ and $w,z \in \set{1,\dots,k}$.
Any of these components includes a dummy request.
\end{enumerate} 
We call such a forest a locality-based forest. Any locality-based forest is denoted by $\forest_{\grd}$.
The following theorem provides a general version of \Cref{th:hstOptimal}.

\begin{theorem}\label{th:genericResult}
Let $I$ denote an instance of \dsms\ that consists of an HST $T$ where the communication is asynchronous and
a set $\R$ of requests that are simultaneously invoked at leaves of $T$.
The protocol \alg\ is optimal if
the total cost of the resulted forest by \alg\ is upper bounded by the total weight of $\forest_{\grd}$. 
\end{theorem}

\subsection{Optimality of GNN on HSTs}
\label{sec:gnnHST}

Consider a one-shot execution of \gnn, and suppose that $\forest_{\gnn}$ is the resulted forest
when running \gnn\ on the given HST $T$ w.r.t. the input sequence $\R$. With respect to \Cref{th:genericResult},
and the fact that \gnn\ only sends one uniques message for scheduling a request to its predecessor,
it is sufficient to show that the forest $\forest_{\gnn}$ can be transformed into a locality-based forest such that the total cost of $\forest_{\gnn}$
is upper bounded by the total weight of $\forest_{\grd}$.
During an execution of \gnn, the Intra or Inter-Component property can be violated (see \Cref{fig:complication}).
Consider the following situations:
\begin{enumerate}
\item A server goes back to a subtree after the time when it leaves the subtree.
\item A request in a subtree of $T$ that initially hosts
at least one server is served by a server that is not initially in the subtree.
\item Two requests in a subtree of $T$ that does not initially host any server, are served by different servers.
\end{enumerate}
The first situation violates the Intra-Component property.
Any of the second and the third situation violates the Inter-Component property.
In the following, we characterize the Intra-Component and the Inter-Component properties by considering a timeline for the
messages that enter and leave a subtree of $T$.
Consider a message $\mu$ that enters the subtree $T'$ of $T$.
Another message can enter $T'$ only after some message $\mu'$ has left $T'$ after $\mu$ entered $T'$---
the arrival times of messages $\mu$ and $\mu'$ at the root of $T'$ can be the same (see \Cref{le:processTime} and \Cref{le:timeLine}).
Similarly, a message can leave $T'$ after $\mu'$ left $T'$ only after some message has entered $T'$ after $\mu'$ left $T'$.
We refer to \Cref{le:timeLine} for more details.
Consider a message $\mu$ that enters $T'$.
The fact that $\mu$ enters $T'$ implies that a server will leave $T'$ for serving $r_{\src}(\mu)$.
Let $\mu'$ denote the first message that leaves $T'$ after $\mu$ entered $T'$.
Leaving $\mu'$ from $T'$ implies that a server will enter $T'$ for serving $r_{\src}(\mu')$.
If $r_{\src}(\mu')$ is in the same TSP path of $\forest_{\gnn}$ with $r_{\des}(\mu)$, then
the server that had served $r_{\des}(\mu)$ goes back to $T'$ for serving $r_{\src}(\mu')$ after it left $T'$, and therefore the Intra-Component property is violated.
Otherwise, the Inter-Component property is violated since two requests in $T'$ are served by two different
servers in which at least one of the servers is initially outside of $T'$. We say \gnn\ makes an \textbf{Inter-Component gap}
$(\mu,\mu')$ on $T'$ in the latter case and an \textbf{Intra-Component gap} $(\mu,\mu')$ on $T'$ in the former case\label{no:gap}.

\para{Transformation:} We transform $\forest_{\gnn}$ through \textit{closing the gaps} that are made by \gnn\ on all subtrees of $T$.
A message $\mu'$ can leave from several subtrees of $T$ such that different messages enter the subtrees before $\mu'$.
Therefore, \gnn\ can make different gaps with the same message $\mu'$ on this set of subtrees of $T$.
We especially refer to \Cref{le:gapLowest} and \Cref{le:gapWindow} for more details on the gaps of the subtrees of $T$.
We consider the lowest subtree in this set and let $(\mu,\mu')$ be a gap on that.
We \textbf{close the gap} $(\mu,\mu')$ by removing $e(\mu')$ and by adding the new edge $\big(r_{\des}(\mu),r_{\src}(\mu')\big)$.
In the example of \Cref{fig:complication}, for instance, the red edges are removed and the new edges $(r^1_0,r_b)$ and $(r_b,r_c)$ are added.
When we close the gap $(\mu,\mu')$, all other gaps $(\mu'',\mu')$ that are on higher subtrees are also closed.
Therefore, we transform $\forest_{\gnn}$ into a new forest $\forest_{mdf}$ by means of closing all gaps\label{no:transformedForest}.
The following lemma shows that $\forest_{mdf}$ is indeed the locality-based forest.

\begin{lemma}\label{le:transformation}
$\forest_{mdf}$ is an $\R_D$-respecting spanning $k$-forest of $B$ that satisfies the Intra-Component
and the Inter-Component properties.
\end{lemma}

It remains to show that the total cost of $\forest_{\gnn}$ is upper bounded by the total weight of the new forest $\forest_{mdf}$. Formally,
we want to show that  $\Latency_{\gnn}(\forest_{\gnn}) \leq \Weight_T(\forest_{mdf})$.
Using \Cref{le:uniquePath}, a message always finds the node of its predecessor
using a direct path on $T$ in any execution of \gnn. Regarding to our communication model described in \Cref{sec:communicationModel},
therefore, for every edge $e \in \forest_{\gnn}$ we have
\begin{equation}\label{eq:latencyUB}
	\latency_{\gnn}\big(\mu(e)\big) \leq \weight_T(e)
\end{equation}
Let $(\mu,\mu')$ be the gap on the lowest subtree of $T$ among all subtrees of $T$ with gaps $(\mu'',\mu')$
for any message $\mu''$ that makes a gap with $\mu'$.
By closing the gap $(\mu,\mu')$, we remove $e^{old}:=\big(r_{\src}(\mu'),r_{\des}(\mu')\big)$ and add the new edge
$e^{new}:=\big(r_{\src}(\mu'),r_{\des}(\mu)\big)$. Using \eqref{eq:latencyUB}, we are immediately done if the latency of $\mu'$ is upper bounded by the weight of
$e^{new}$. However, the latency of $\mu'$ can be larger than the weight of $e^{new}$.
By contrast, the weight of $e^{new}$ is lower bounded by the latency of $\mu$
(see \Cref{co:gapLowest} and \Cref{le:greedyNature}). This lower bound gives us the go-ahead to show that the weight of $e^{new}$ can
be seen as an ``amortized'' upper bound for $\latency_{\gnn}(\mu')$. In the following, we provide \textbf{an overview of our amortized
analysis} that appears in \Cref{sec:gnnUB}. 
Let $E^{new}:=\forest_{mdf} \setminus \forest_{\gnn}$ and $E^{old}:=\forest_{\gnn} \setminus \forest_{mdf}$
be the sets of all edges that are added and removed during the transformation of $\forest_{\gnn}$, respectively\label{no:newOldEdges}.
Further, we consider a set of edges that provides enough ``potential'' for our amortization\label{no:potEdges}. 
\[ 
	E^{pot}:=\set{e \in \forest_{\gnn} :\big(\mu(e),\mu(e')\big) \ \text{is a gap for some} \ e' \in E^{old}}.
\]
For every edge $e \in E^{old}$, let $E^{pot}(e):=\set{e' \in E^{pot} : \big(\mu(e'),\mu(e)\big) \ \text{is a gap}}$\label{no:filterEPOT}.
Further, for every edge $e \in E^{pot}$, let $E^{old}(e):=\set{e' \in E^{old} : \big(\mu(e),\mu(e')\big) \ \text{is a gap}}$\label{no:filterEOLD}.
In this overview, we consider the \textbf{simple case} where 1) $|E^{old}(e)|=1$ for every edge
$e \in E^{pot}$ and $|E^{pot}(e)|=1$ for every edge $e \in E^{old}$. Further, 2) the sets
$E^{old}$ and $E^{pot}$ do not share any edge.
The execution provided by \Cref{fig:complication} represents an example of the above simple case.
We define the \textbf{potential function} $\Phi(F)$ for a subset $F$ of $\forest_{\gnn}$ as follows\label{no:potential}
$\Phi(F) := \Weight_T(F) - \Latency_{\gnn}(F)$.
W.l.o.g., we assume that the edges in $E^{old}$ are sequentially replaced with the edges in $E^{new}$.
Hence, assume that $e^{old}_i$ is replaced with $e^{new}_i$ during the $i$-th replacement. Let also $e^{pot}_i$ be the only edge in $E^{pot}(e^{old}_i)$.

\begin{lemma}\label{le:amortizedSimple}
If $|E^{old}(e)|=1$ for every edge $e \in E^{pot}$, $|E^{pot}(e)|=1$ for every edge $e \in E^{old}$, and $E^{old} \cap E^{pot} = \emptyset$,
then 
\begin{equation}\label{eq:amortizedSimple}
	\weight_T(e^{old}_i) \leq \weight_T(e^{new}_i) + \Phi\left(E^{pot}\setminus \set{e^{pot}_1,\dots,e^{pot}_{i-1}}\right) -
	\Phi\left(E^{pot}\setminus \set{e^{pot}_1,\dots,e^{pot}_i}\right)
\end{equation}
for every $i \geq 1$.
\end{lemma}
\begin{proof}
Using the definition of the potential function $\Phi$ and the definitions of the total weight and the total communication cost
of a subset of edges in $\forest_{\gnn}$,
we have
\[
	\Phi\left(E^{pot}\setminus \set{e^{pot}_1,\dots,e^{pot}_{i-1}}\right) - \Phi\left(E^{pot}\setminus \set{e^{pot}_1,\dots,e^{pot}_i}\right)
	=\weight_T(e^{pot}_i) - \latency_{\gnn}(e^{pot}_i).
\]
Therefore, we need to show that
$\weight_T(e^{old}_i) \leq \weight_T(e^{new}_i) + \weight_T(e^{pot}_i) - \latency_{\gnn}(e^{pot}_i)$.
Let the subtree $T'$ of $T$ be the lowest subtree such that $\big(\mu(e^{pot}_i),\mu(e^{old}_i)\big)$ is a gap on $T'$. This implies that
$\weight_T(e^{new}_i)=\delta(T')$. On the other hand, using \Cref{le:greedyNature} we have
$\latency_{\gnn}(e^{pot}_i) \leq \delta(T')=\weight_T(e^{new}_i)$.
It remains to show that $\weight_T(e^{old}_i) \leq \weight_T(e^{pot}_i)$.
Let $T''_j$ be the highest subtree of $T$ such that $\big(\mu(e^{pot}_i),\mu(e^{old}_i)\big)$ is a gap on $T''_j$ and $T''_j$ is
a child subtree of $T''$. The message $\mu(e^{old}_i)$ does not leave $T''$ since $E^{pot}(e^{old}_i)=\set{e^{pot}_i}$.
Hence, $\weight_T(e^{old}_i) = \delta(T'')$. On the other hand, the fact that the message $\mu(e^{pot}_i)$ enters $T''_j$ indicates that
$\weight_T(e^{pot}_i) \geq \delta(T'')$. Consequently, $\weight_T(e^{pot}_i) \geq \weight_T(e^{old}_i)$ and we are done.
\end{proof}
When we sum up \eqref{eq:amortizedSimple} for all $i$, we get
\begin{equation}\label{eq:cumulativeAmortization}
	\Weight_T(E^{old}) \leq \Weight_T(E^{new}) + \Phi\left(E^{pot}\right).
\end{equation}
Using the definition of the potential function $\Phi$ and using $\Latency_{\gnn}(E^{old}) \leq \Weight_T(E^{old})$ w.r.t \eqref{eq:latencyUB},
therefore we get $\Latency_{\gnn}(E^{pot}) + \Latency_{\gnn}(E^{old}) \leq \Weight_T(E^{new}) + \Weight_T(E^{pot})$.
Hence, we have $\Latency_{\gnn}(\forest_{\gnn}) \leq \Weight_T(\forest_{mdf})$ since $\forest_{mdf}=\forest_{\gnn} \setminus E^{old} \cup E^{new}$ and
$\Latency_{\gnn}\big(\forest_{\gnn} \setminus (E^{old} \cup E^{pot})\big) \leq \Weight_T\big(\forest_{\gnn} \setminus (E^{old} \cup E^{pot})\big)$
 w.r.t \eqref{eq:latencyUB}.

\begin{lemma}\label{le:UBGNNForest}
The total cost of $\forest_{\gnn}$ is upper bounded by the total weight of $\forest_{mdf}$.
\end{lemma}

\begin{theorem}\label{th:gnn}
The forest $\forest_{\gnn}$ can be transformed into the locality-based forest $\forest_{\grd}$ such that the total cot of $\forest_{\gnn}$
is upper bounded by the total weight of $\forest_{\grd}$. 
\end{theorem}

\section{Analysis}
\label{sec:analysis} 
In this section, we provide a complete version of our analysis given in \Cref{sec:compactanalysis}.
At the end of this section, we consider the general graphs and show that the claim of \Cref{th:graph} holds. 
We use the abbreviation $[i]:=\set{0,\dots,i}$ and $[i,j]:=\set{i,\dots,j}$ for non-negative integers $i$ and $j$.
Recall from \Cref{sec:compactanalysis} that \alg\ is a particular distributed \dsms\ protocol
that sends a unique message from the node of a request to the node of the predecessor request for scheduling the request. 
In a one-shot execution of \alg, the total communication cost of the schedule $\pi^z_{\alg}$ for any $z \in [1,k]$ equals the
total communication cost of the corresponding TSP path in the resulted forest by \alg\ w.r.t. \eqref{eq:componentTotalCost} and \eqref{eq:totalLatencyF}.
\begin{equation}\label{eq:componentLatencyDelay}
	\Delay_{\alg}(\pi^z_{\alg})=\Latency_{\alg}(F^z_{\alg}).
\end{equation}
Using \eqref{eq:totalCost} and \eqref{eq:componentLatencyDelay}, therefore, the total cost of \alg\
equals the total cost of the resulted forest $\forest_{\alg}$ that is sum of the total costs of the $k$ TSP paths.

\begin{equation}\label{eq:forestAlgWeight}
	\Delay_{\alg}=\Latency_{\alg}(\forest_{\alg}) := \sum_{z=1}^{k}\Latency_{\alg}(F^z_{\alg}).
\end{equation}

\subsection{Optimal Distributed \dsms\ Protocols}
\label{sec:optimalProtocol}
When studying the cost of an optimal offline \dsms\ protocol \opt, we assume that \opt\ knows the whole sequence of requests in advance. However, \opt\ still
needs to send messages from each request to its predecessor request. In \Cref{sec:communicationModel}, we explained that the message latencies are not
even under control of an optimal distributed \dsms\ protocol, denoted by \opt.

\begin{remark}\label{re:asynchOpt}
For lower bounding the cost of an optimal protocol that solves a distributed problem, one can assume that all communication is
synchronous even in an asynchronous execution since a synchronous execution is a possible strategy of the asynchronous scheduler.
\end{remark} 

For lower bounding the total cost of \opt, we assume that all communication is synchronous w.r.t. \Cref{re:asynchOpt}.
Let $\forest_{\opt}$ denote the resulted forest by \opt\ in a synchronous execution that includes
$k$ TSP paths denoted by $F^1_{\opt},F^2_{\opt}, \dots, F^k_{\opt}$.
Note that \opt\ only sends one message for scheduling of a request.
Regarding to \eqref{eq:delayOneshot},
the scheduling cost of any request $r'=(v',0)$ as the successor request of $r=(v,0)$ equals $d_G(v,v')$ for the input graph $G=(V,E)$
in a synchronous system w.r.t. \Cref{sec:communicationModel} if the corresponding find-predecessor
message is sent through a direct path from $v'$ to $v$.
Therefore, the total cost of \opt\ is at least the total weight of $\forest_{\opt}$ w.r.t measurements
of the input graph. Formally,

\begin{equation}\label{eq:optAlgCostLB}
	\Delay_{\opt} \geq \Weight_G(\forest_{\opt}).
\end{equation}

\begin{remark}\label{re:resultedForestAlg}
With respect to \Cref{de:SRespectingLForest}, the resulted forest by any distributed protocol that solves
the \dsms\ problem is an $\R_D$-respecting spanning $k$-forest of $B=\big(\R,{\R \choose 2}\big)$.       
\end{remark}

Regarding to \Cref{re:resultedForestAlg}, the total weight of $\forest_{\opt}$ can be lower bounded
by the total weight of a minimum $\R_D$-respecting spanning $k$-forest of $B=\big(\R,{\R \choose 2}\big)$
w.r.t. the measurements of the input graph $G$. Let $\forest_{\min}$ denote any minimum weight
$\R_D$-respecting spanning $k$-forest of $B$ that includes
$k$ trees denoted by $F^1_{\min},F^2_{\min}, \dots, F^k_{\min}$\label{no:minForest}. Formally,

\begin{equation}\label{eq:optForestCostLB}
	\Weight_G(\forest_{\opt}) \geq \Weight_G(\forest_{\min}).
\end{equation}

\subsection{Optimal Distributed \dsms\ Protocols on HSTs}
\label{sec:optimalProtocolHST}
We consider the HST $T$
(see \Cref{de:hst}) as the input graph in this section. A set of requests $\R$ including the $k \geq 1$ dummy
requests $\R_D$ as well as $k$ identical servers are initially located at some leaves of $T$. As explained in
the problem definition in \Cref{sec:dsms}, we assume w.l.o.g. that there is a dummy request
at each leaf of $T$ that initially hosts a server.
We observe that if \alg\ satisfies the following basic \textbf{locality properties} on $T$,
then it outputs an optimal solution w.r.t. $T$. Later, we will formally show that this observation is
indeed correct.

\begin{enumerate}[1.] 
\item \label{pr:ALGIntracomponent} A server never goes back to a subtree after the time
when it leaves the subtree.
\item \label{pr:ALGIntercomponent1} All requests in any subtree of $T$ that initially hosts
at least one server, are served by the servers inside the subtree.
\item \label{pr:ALGIntercomponent2} All requests in any subtree of $T$ that does not initially host any server,
are served by the same server.
\end{enumerate} 

However, w.r.t. \eqref{eq:optForestCostLB} we would like to characterize the properties
of $\forest_{\min}$ instead of working with $\forest_{\opt}$. While any component
of $\forest_{\opt}$ is a TSP path that represents the movement of the corresponding server,
any component of $\forest_{\min}$ is a tree and not necessarily a TSP path.
Hence, Property \ref{pr:ALGIntracomponent}, Property \ref{pr:ALGIntercomponent1}, and Property \ref{pr:ALGIntercomponent2} are adapted
for any $\R_D$-respecting spanning $k$-forest of $B$ that has a minimum total weight as follows.

Let the $k$ components of any $\R_D$-respecting spanning $k$-forest $\forest$ of $B$
denoted by $F^1,F^2, \dots, F^k$. Further, let $\R^w$ denote the request set of the component
$F^w$ for any $w \in [1,k]$. For any subtree $T'$ of $T$, let $F^w(T')$ denote the subgraph of
$F^w$ that is induced by those requests in $\R^w$ that are also in $T'$ and let $\R^w(T')$
denote the request set of the component $F^w(T')$. Let $\forest_{\grd}$ denote any $\R_D$-respecting
spanning $k$-forest of $B$ that has the following basic \textbf{locality properties}.

\begin{enumerate}[I.] 
\item \label{pr:ForestIntracomponent} [\textbf{Intra-Component Property}] For any subtree $T'$ of $T$
and for any $w \in [1,k]$, the component $F^w_{\grd}(T')$ is a tree.
The Intra-Component property is adapted from the Property \ref{pr:ALGIntracomponent}.
\item \label{pr:ForestIntercomponent} [\textbf{Inter-Component Property}] For any subtree $T'$ of $T$, suppose that
there are at least two non-empty components $F^z_{\grd}(T')$ and $F^w_{\grd}(T')$ where $w \neq z$ and $w,z \in [1,k]$.
Any of these components includes a dummy request.
This property combines Property \ref{pr:ALGIntercomponent1} and Property \ref{pr:ALGIntercomponent2}
into one property.
\end{enumerate} 

Note that the Property \ref{pr:ForestIntercomponent} implies that for any subtree $T'$ of $T$ that does not
initially host a server, all requests in $T'$ are included in a unique component $F^w_{\grd}(T')$ for some $w \in [1,k]$.
We generally say that removing an edge from any graph $G=(V,E)$ provides an
$m$-cut if the removal decomposes the graph into $m$ connected components.
The following lemma elaborates more formally the Property
\ref{pr:ForestIntracomponent} and Property \ref{pr:ForestIntercomponent}.    
 
\begin{lemma}\label{le:localPropertiesAnalysis}
For every edge $e$ in $\forest_{\grd}$, consider the shortest weight edge $e^*$ in
$E_B$ crossing the $(k+1)$-cut induced by removing $e$ from $\forest_{\grd}$ such
that $\forest_{\grd} \setminus \set{e} \cup \set{e^*}$ is again an $\R_D$-respecting spanning
$k$-forest of $B$. Then, the weight of $e$ equals the weight of $e^*$.
\end{lemma}
\begin{proof}
Let $e=(r_p,r_q)$ be in the component $F^w_{\grd}$ for some $w \in [1,k]$. We consider two different cases
as follows. 1) The edge $e^*$ crosses the cut---or $2$-cut w.r.t. our definition of an $m$-cut---over $F^w_{\grd}$ resulted by removing $e$.
2) The edge $e^*$ crosses the $(k+1)$-cut over $\forest_{\grd}$ resulted by removing $e$
and connects two different components of $\forest_{\grd}$.

When we remove $e$ from $F^w_{\grd}$, let $(\R^{w_1}_{\grd},\R^{w_2}_{\grd})$ be the resulted
$2$-cut. First, we show that the weight of $e^*=(r_x,r_y)$ is not smaller than weight of $e$
in the first case using the Property \ref{pr:ForestIntracomponent}. 
Regarding to the Property \ref{pr:ForestIntracomponent}, the
tree $F^w_{\grd}(T')$ is connected for every subtree $T'$ of $T$.
Let $T'$ be the lowest subtree of $T$ that consists of both $r_p$ and $r_q$ as
shown in \Cref{fig:intracomponentEdgeAnalysis}. Hence, removing $e$ from $F^w_{\grd}(T')$
does not decompose the tree $F^w_{\grd}(T'')$ for any subtree $T''$ of $T$
whose root is not an ancestor of the root of $T'$---in this paper, we assume that a node is an ancestor of itself.
Thus, the request set $\R^w_{\grd}(T'')$ is completely either in $\R^{w_1}_{\grd}$ or $\R^{w_2}_{\grd}$
(this decomposition is shown with blue and gray colors in \Cref{fig:intracomponentEdgeAnalysis}).
Therefore, the minimum depth of $T''$ is at least the depth of a child subtree of $T'$.   
We recall that the two requests $r_x$ and $r_y$ are on different sides of
the cut $(\R^{w_1}_{\grd},\R^{w_2}_{\grd})$. Hence, the least common ancestor
of $v_x$ and $v_y$ has to be an ancestor of $T'$. Thus, we get $d_T(v_p,v_q) \leq d_T(v_x,v_y)$.

\begin{figure}[H]
  \center	
  \includegraphics[width=0.4\textwidth]{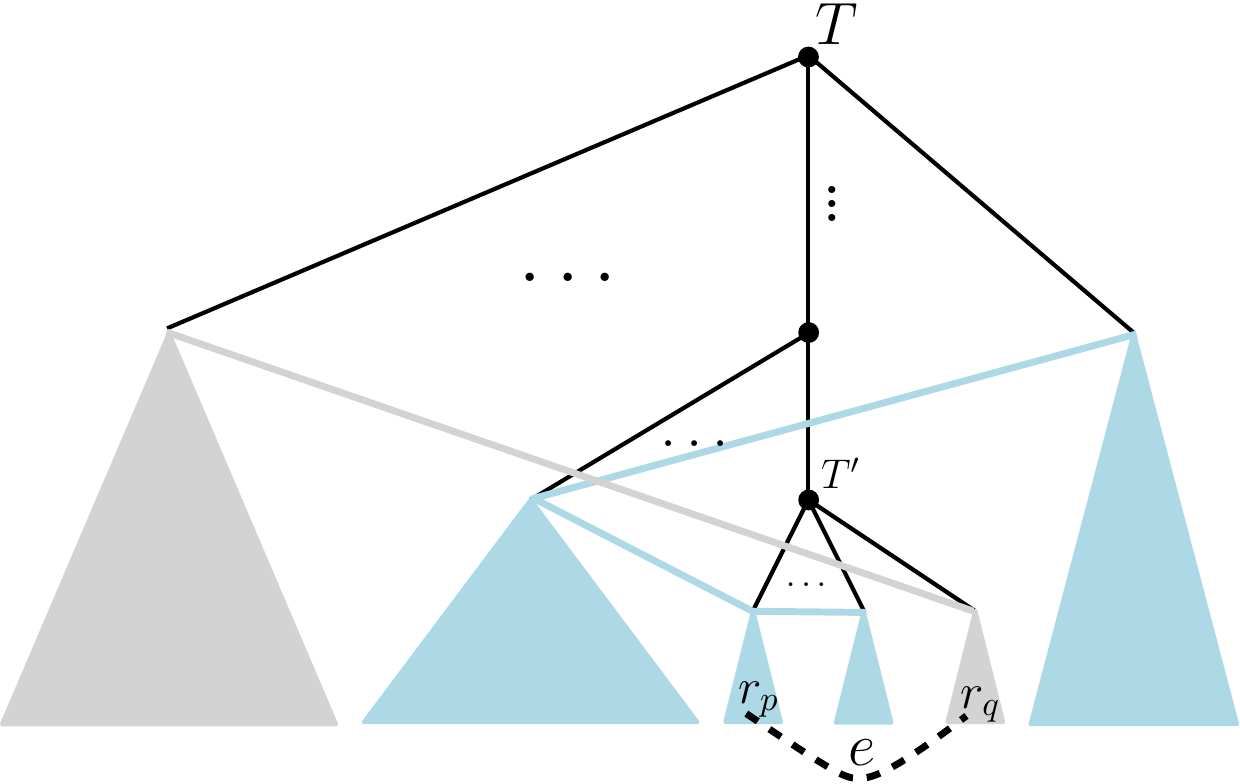}
  \caption{The component $F^w_{\grd}$:
  The connected component corresponding to the request set $\R^{w_1}_{\grd}$ is colored with blue color and
  the one corresponding to the request set $\R^{w_2}_{\grd}$ is colored with gray color. For any edge
  $e^*=(r_x,r_y)$ that connects two components with different colors, the requests $r_x$ and $r_y$ both
  together can be only included in a subtree whose depth is not smaller than the depth of $T'$.}
  \label{fig:intracomponentEdgeAnalysis}
\end{figure}

Now, we show that the weight of $e^*=(r_x,r_y)$ is not smaller than the weight of $e$
in the second case using the Property \ref{pr:ForestIntercomponent}. W.l.o.g., assume that the dummy
request is in $\R^{w_1}_{\grd}$. Hence, one of the endpoints of $e^*$ must be
in $\R^{w_2}_{\grd}$ and the other endpoint must belong to some different component $F^z_{\grd}$.
Let the subtree $T'$ of $T$ be the lowest subtree that consists of both $r_x$ and $r_y$.
For the sake of contradiction, assume that the weight of $e^*$ is smaller than the weight
of $e$. Thus, the request set $\R^w_{\grd}(T')$ that has one of the endpoints
of $e^*$, does not have any dummy request since $T'$ is the lowest subtree that
includes $e^*$, and the weight of $e$ is larger than $e^*$.
However, using the Property \ref{pr:ForestIntercomponent}, the component $F^w_{\grd}(T')$
must include a dummy request since there is another component $F^z_{\grd}(T')$ that includes
another endpoint of $e^*$. This is a contradiction, and therefore the weight of $e^*$ is not
smaller than the weight of $e$.
\end{proof}
\begin{lemma}\label{le:optimalGreedyForest}
Any $\R_D$-respecting spanning $k$-forest of $B$ that has the Property
\ref{pr:ForestIntracomponent} and Property \ref{pr:ForestIntercomponent}
is indeed a minimum weight $\R_D$-respecting spanning $k$-forest of $B$. Consequently,
\[
	\Weight_T(\forest_{\min}) = \Weight_T(\forest_{\grd}).
\]
\end{lemma}
\begin{proof}
Together with \Cref{le:localPropertiesAnalysis}, \Cref{appthm:MSFapprox} implies that the total
weight of $\forest_{\grd}$ equals the total weight of $\forest_{\min}$.
\end{proof}

In an asynchronous system, the message latencies are unpredictable. Hence, the resulted
forest by some optimal \dsms\ protocol on $T$ might violate the Property
\ref{pr:ALGIntracomponent}, Property \ref{pr:ALGIntercomponent1}, and Property \ref{pr:ALGIntercomponent2}.

\begin{proof}[\textbf{Proof of \Cref{th:genericResult}}]
Let $\forest_{\alg}$ denote the resulted forest by \alg.
Using the assumptions of the theorem, we have
\[
	\Delay_{\alg}=\Latency_{\alg}(\forest_{\alg}) \leq \Weight_T(\forest_{\grd}).
\]
On the other hand, using \eqref{eq:optAlgCostLB}, \eqref{eq:optForestCostLB}, and \Cref{le:optimalGreedyForest}
we have
\[
	\Delay_{\opt} \geq \Weight_T(\forest_{\opt}) \geq \Weight_T(\forest_{\grd}).
\]
\end{proof}

\subsection{Optimality of GNN on HSTs}
\label{sec:gnnHST}

In this section, we provide a proof for \Cref{th:gnn}.
In an execution of \gnn\ protocol, we can have the situation where the Property
\ref{pr:ALGIntracomponent}, Property \ref{pr:ALGIntercomponent1}, and Property \ref{pr:ALGIntercomponent2} are violated. Consider an example provided by \Cref{fig:complication}.
In a one-shot execution of \gnn, there are initially two dummy $r^1_0$ and $r^2_0$ requests and four
other requests $r_a,r_b,r_c,r_d$. The result of this execution is as follows.
The find-predecessor message of $r_a$ reaches the root of $T''$ not later than the time when the
find-predecessor message of $r_b$ reaches the root of $T''$. This implies that the
server that is initially in $T''$ leaves $T''$ for serving the request $r_a$ and it goes back again into $T''$
for serving $r_b$. Therefore, the Property \ref{pr:ALGIntracomponent}---or the Intra-Component property---is violated.
Further, the find-predecessor message of $r_d$ reaches the root of $T'$ not earlier than the
find-predecessor messages of $r_a$ and $r_b$ and reaches the root of $T''$ not later than the
find-predecessor messages of $r_c$. Hence, the server that is not initially in $T''$ enters $T''$
for serving $r_c$. Therefore, the Property \ref{pr:ALGIntercomponent1}---or the Inter-Component property---is violated.

\subsubsection{Transforming $\forest_{GNN}$}
\label{sec:modification}

Although $\forest_{\gnn}$ does not necessarily satisfy the Intra-Component and Inter-Component properties
as it is shown by \Cref{fig:complication} as an example, the ``greedy nature'' of \gnn\ helps us to transform $\forest_{\gnn}$ into an
$\R_D$-respecting spanning $k$-forest of $B$ that satisfies the Property
\ref{pr:ForestIntracomponent} and Property \ref{pr:ForestIntercomponent} where $\R_D \subseteq \R$
such that the total cost of $\forest_{\gnn}$ is upper bounded by the total weight of the new
forest. Before we delve into details of the transformation, some preliminaries are provided in the following.

Let $M^{\uparrow}(T'):=\{\mu^{\uparrow}_1(T'),\mu^{\uparrow}_2(T'),\dots\}$ denote the set of all
find-predecessor messages of requests in $T'$ that leave $T'$\label{no:messagesUP}.
We assume that $\mu^{\uparrow}_1(T')$ is the first message that leaves $T'$.
In general, the messages in $M^{\uparrow}(T')$ are indexed in the order they leave $T'$. 
We also consider the set of all find-predecessor messages that enter the subtree $T'$.
Let $M^{\downarrow}(T'):=\{\mu^{\downarrow}_1(T'),\mu^{\downarrow}_2(T'),\dots\}$ denote the
set of all find-predecessor messages of requests in $T'$ that enter the subtree $T'$\label{no:messagesDown}.
The message $\mu^{\downarrow}_1(T')$ is assumed to be the first message that enters $T'$.
Similarly, the messages in $M^{\downarrow}(T')$ are indexed in the order they enter $T'$.
In the following, we provide a timeline from the first time when a message
that is either in $M^{\uparrow}(T')$ or in $M^{\downarrow}(T')$ reaches the root of $T'$ for
any subtree $T'$ of $T$ to the time when the last message among all messages in these two sets
reaches the root of $T'$.

If the subtree $T'$ only includes non-dummy requests, then the first message that reaches
the root of $T'$ must leave $T'$ w.r.t. the description of \gnn. By contrast, if $T'$ includes a
dummy request, w.l.o.g., we assume that $\mu^{\uparrow}_1(T')$ is a ``virtual message''
that reaches the root of $T'$ at time $0$. Therefore, in either case, $M^{\uparrow}(T')$
includes $\mu^{\uparrow}_1(T')$ and $|M^{\uparrow}(T')| \geq 1$. The set $M^{\downarrow}(T')$, however,
can be empty if there is not any request in $T'$ that has a successor request outside of $T'$. 
Regarding the description of \gnn, note that the root of $T'$ arbitrarily processes the messages
that are received at the same time and \gnn\ does not take into account the processing time
at the root of $T'$.    

\begin{lemma}\label{le:processTime}
Consider any subtree $T'$ of $T$ and any two messages that reach the root $v$ of $T'$ at the same time $t$
in which one of them is from inside of $T'$ denoted by $\mu$ and the other one is from outside of $T'$ denoted by $\mu'$.
The node $v$ processes $\mu'$ before processing the message $\mu$ if the message $\mu$ leaves $T'$.
\end{lemma}
\begin{proof}
The node $v$ does not point to its parent at time $t$ before it processes the message $\mu'$. Further, at the same
time, $v$ must point to at least one of its children w.r.t. \Cref{le:uniqueEdge}. 
Therefore, if $v$ processes $\mu$ before processing the message $\mu'$, then $\mu$ is forwarded inside $T'$
w.r.t. the description of \gnn.
This is a contradiction with the assumption of the lemma
and consequently $v$ processes $\mu'$ before processing the message $\mu$.
\end{proof}

\begin{lemma}\label{le:upwardArrow}
Consider any subtree $T'$ of $T$ in which $M^{\downarrow}(T')$ is not empty.
The root $v$ of $T'$ has an upward link since when $\mu^{\downarrow}_i(T')$ reaches $v$ and it
is processed by $v$ until the first time after processing $\mu^{\downarrow}_i(T')$
when a find-predecessor message that leaves $T'$ is processed by $v$
for any $i \in [1,|M^{\downarrow}(T')|]$.
\end{lemma}
\begin{proof}
Let $\mu^{\uparrow}_j(T')$ denote the first message that leaves $T'$ not earlier than $t^{\downarrow}_i(T')$.
Hence, we have $t^{\uparrow}_j(T') \geq t^{\downarrow}_i(T')$.
The message $\mu^{\downarrow}_i(T')$ is processed by $v$ before processing $\mu^{\uparrow}_j(T')$ 
even if $\mu^{\downarrow}_i(T')$ and $\mu^{\uparrow}_j(T')$ reach $v$ at the same time using \Cref{le:processTime}.
Therefore, $v$ forwards $\mu^{\downarrow}_i(T')$ to $T'$ and points to its parent w.r.t. the description of \gnn.
The node $v$ has the pointer to its parent until when $v$ processes $\mu^{\uparrow}_j(T')$. The node $v$ forwards
$\mu^{\uparrow}_j(T')$ to its parent and the link that points to its parent is removed.
\end{proof}

\begin{lemma}\label{le:DownwardArrow}
Consider any subtree $T'$ of $T$ that includes at least one request.
The root $v$ of $T'$ has a downward link since when $\mu^{\uparrow}_i(T')$ reaches $v$ and it
is processed by $v$ until the first time $t$ after processing $\mu^{\uparrow}_i(T')$
when a find-predecessor message that enters $T'$ is processed by $v$ for any $i \in [1,|M^{\uparrow}(T')|]$.
Further, $t$ is strictly larger than $t^{\uparrow}_i(T')$.
\end{lemma}
\begin{proof}
Let $\mu^{\downarrow}_j(T')$ denote the first find-predecessor message that is processed by $v$ after processing $\mu^{\uparrow}_i(T')$.
The messages $\mu^{\downarrow}_j(T')$ and $\mu^{\uparrow}_i(T')$ cannot reach $v$ at the same time as otherwise
$\mu^{\downarrow}_j(T')$ is processed before $\mu^{\uparrow}_i(T')$ by $v$ using \Cref{le:processTime}. Hence, we have
$t^{\uparrow}_i(T') < t=t^{\downarrow}_j(T')$.

If $t^{\uparrow}_i(T')=0$, then there is a downward directed path from $v$ at the beginning of the execution since there is a dummy
request in $T'$ and therefore $v$ has a downward link at time $t^{\uparrow}_i(T')$.
Otherwise, $t^{\uparrow}_i(T')$ must be larger than $0$. Suppose that the child node
$u$ of $v$ sends the message $\mu^{\uparrow}_i(T')$ to $v$.
The node $v$ atomically forwards the message $\mu^{\uparrow}_i(T')$ at time $t^{\uparrow}_i(T')$ to its parent and set a link that points from $v$
to $u$ w.r.t. the description of \gnn. 
Therefore, in either case, $v$ has a downward link at time $t^{\uparrow}_i(T')$.
The node $v$ points to some other child $w$ if it receives another find-predecessor message from
$w$ and forwards it either to $u$ or some other child node. Hence, $v$ has a downward link since $t^{\uparrow}_i(T')$ as long as
$v$ receives a find-predecessor message of some request that is inside $T'$. As soon as $\mu^{\downarrow}_j(T')$ reaches
$v$, the node $v$ forwards it to $T'$. If $v$ has only one downward link immediately before $\mu^{\downarrow}_j(T')$ reaches $v$,
the downward link is removed when the message $\mu^{\downarrow}_j(T')$ is forwarded to $T'$.
Therefore, the node $v$ has a downward link from time $t^{\uparrow}_i(T')$ until at least $t^{\downarrow}_j(T')$.
\end{proof}

Finally, we are ready to provide a timeline from the time when the first message from the sets $M^{\downarrow}(T')$ and
$M^{\uparrow}(T')$ reaches the root of $T'$ to the time when the last message among all messages in these two sets reaches
the root of $T'$.

\begin{lemma}\label{le:timeLine}
Consider any subtree $T'$ of $T$ that includes at least one request. Let $v$ denote the root of $T'$.
We have
\[
	t^{\uparrow}_1(T') < t^{\downarrow}_1(T') \leq
	t^{\uparrow}_2(T') < t^{\downarrow}_2(T') \leq 
	t^{\uparrow}_3(T') < \cdots.
\]
\end{lemma}
\begin{proof}
If the subtree $T'$ includes only non-dummy requests, then the first message that reaches $v$ must be $\mu^{\uparrow}_1(T')$
w.r.t. the description of \gnn. Even if the subtree $T'$ includes a dummy request, then $\mu^{\uparrow}_1(T')$ is again
the first message that reaches $v$ since $\mu^{\uparrow}_1(T')$ is a virtual message in this case and $t^{\uparrow}_1(T')=0$.

First, we show that
\[
	t^{\uparrow}_1(T') < t^{\downarrow}_1(T') \leq t^{\uparrow}_2(T').
\]
We know that if $\mu^{\uparrow}_1(T')$ is a virtual message, then $t^{\uparrow}_1(T')=0$.
Therefore, the message $\mu^{\downarrow}_1(T')$ can only reach $v$ after time $0$ in this case.
Otherwise, the subtree $T'$ does not include any dummy request if $\mu^{\uparrow}_1(T')$ is an actual
message. Therefore, w.r.t. the initialization of \gnn, the root
$v$ of $T'$ has an upward link since there is not any dummy request in $T'$ 
until time $t^{\uparrow}_1(T')$ when $\mu^{\uparrow}_1(T')$ reaches $v$.
Hence, the messages $\mu^{\downarrow}_1(T')$ can only reach $v$ after $t^{\uparrow}_1(T')$
since $v$ has an upward link immediately before $t^{\uparrow}_1(T')$ and the message $\mu^{\downarrow}_1(T')$
cannot be in transit on that.
Consequently, in either case the message $\mu^{\downarrow}_1(T')$ can only reach $v$ after $t^{\uparrow}_1(T')$. 
On the other hand, $\mu^{\downarrow}_1(T')$ is the first message among all messages in $M^{\downarrow}(T')$
that reaches $v$ and therefore $v$ has a downward link from $t^{\uparrow}_1(T')$ to $t^{\downarrow}_1(T')$
using \Cref{le:DownwardArrow}. Therefore, $\mu^{\uparrow}_2(T')$ cannot reach $v$ before $t^{\downarrow}_1(T')$ as otherwise
it cannot leave $T'$. Consequently, $t^{\uparrow}_1(T') < t^{\downarrow}_1(T') \leq t^{\uparrow}_2(T')$. 

Assume that the following holds,
\[
	t^{\uparrow}_1(T') < t^{\downarrow}_1(T') \leq t^{\uparrow}_2(T') <
	\cdots \leq t^{\uparrow}_{j-1}(T') < t^{\downarrow}_{j-1}(T') \leq t^{\uparrow}_{j}(T').
\]
We now show that the claim of lemma also holds for the following.
\[
	t^{\uparrow}_1(T') < t^{\downarrow}_1(T') \leq t^{\uparrow}_2(T') <
	\cdots \leq t^{\uparrow}_{j}(T') < t^{\downarrow}_{j}(T') \leq t^{\uparrow}_{j+1}(T').
\]
The message $\mu^{\downarrow}_{j}(T')$ can only reach $v$ after $t^{\uparrow}_j(T')$ since $v$ has an upward link
from $t^{\downarrow}_{j-1}(T')$ to $t^{\uparrow}_{j}(T')$ using \Cref{le:upwardArrow}. On the other hand, $\mu^{\uparrow}_{j+1}(T')$
cannot reach $v$ before time $t^{\downarrow}_j(T')$ as otherwise $\mu^{\uparrow}_{j+1}(T')$ finds a downward link at $v$
using \Cref{le:DownwardArrow} and therefore it cannot leave $T'$. Consequently, the calim
of the lemma holds.
\end{proof}

\begin{corollary}\label{co:timeLine}
Consider any subtree $T'$ of $T$ that includes at least one request. We have 
\[
	|M^{\downarrow}(T')| \leq |M^{\uparrow}(T')| \leq |M^{\downarrow}(T')|+1.
\]
\end{corollary}
\begin{proof}
The claim of the corollary holds using \Cref{le:timeLine} since $\mu^{\uparrow}_1(T') \in M^{\uparrow}(T')$.
\end{proof}

\para{Gaps:} Consider any subtree $T'$ of $T$. Assume that $|M^{\uparrow}(T')|=m$.
We say that $\forest_{\gnn}$ has the gap $\big(\mu^{\downarrow}_i(T'),\mu^{\uparrow}_{i+1}(T')\big)$
on $T'$ for any $i \in [1,m-1]$.

We may refer to a gap as an \textbf{Intra-Component} gap if the messages of the gap
are corresponding with the requests that are in the same component of $\forest_{\gnn}$.
If they are corresponding with the requests
that are in two different components of $\forest_{\gnn}$, then the gap is called an \textbf{Inter-Component} gap.
In fact, the Intra-Component and Inter-Component properties are violated when the Intra-Component and Inter-Component
gaps are made, respectively.

Let $\delta(T')$ denote the diameter of any subtree $T'$ of $T$ that is the weight of the longest
direct path between any two leaves in $T'$\label{no:treeDiameter}. Further, let $v_{\src}(\mu)$ and $v_{\des}(\mu)$
the source and destination of the message $\mu$, respectively. 
\textit{In general, when we refer to the gap $(\mu,\mu')$, we assume that
it is the gap on any subtree $T'$ of $T$ in which $\mu$ enters $T'$ as the $i$-th message in $M^{\downarrow}(T')$
and $\mu'$ leaves $T'$ as the $(i+1)$-th message in $M^{\uparrow}(T')$ for any $i \geq 1$.} In the following, we characterize
the situations where a subtree of $T$ has a gap.

\para{Size of a gap:} The size of the gap $(\mu,\mu')$ equals $d_T\big(v_{\des}(\mu),v_{\src}(\mu')\big)$. 

We say that a subtree $T'$ of $T$ is \textbf{isolated} if no find-predecessor message can enter or exit $T'$.
\begin{lemma}\label{le:isolatedTree}
A subtree $T'$ of $T$ is isolated if the root of $T'$ has an upward and a downward links.
\end{lemma}
\begin{proof}
Let $v$ denote the root of $T'$. The upward link implies that no find-predecessor
message can enter $T'$ as long as $v$ has an upward link. On the other hand,
when the first find-predecessor message reaches $v$ from some node $u$, a new downward
link from $v$ to $u$ is created while the find-predecessor is forwarded to an already existed downward link w.r.t. the
\gnn\ protocol. The same happens for all other find-predecessor messages when reaching $v$ afterward.
Thus, $v$ always has a downward link, and the upward link remains unchanged.
Consequently, $T'$ is isolated since any other find-predecessor message cannot enter or exit $T'$.
\end{proof}

\begin{lemma}\label{le:gapLowest}
Consider the lowest subtree $T'$ of $T$ in which $\forest_{\gnn}$ has the gap $(\mu_i,\mu_j)$ on $T'$. The message
$\mu_i$ does not visit any node on the path between $v_{\src}(\mu_j)$ and the root of $T'$ (excluding the root of $T'$).
\end{lemma}
\begin{proof}
Let $v'$ denote the root of $T'$ and $T'_z$ be the child subtree of $T'$ that includes the leaf node $v_{\src}(\mu_j)$.
Further, let $v'_z$ be the root of $T'_z$.
For the sake of contradiction, assume that $\mu_i$ reaches $v'_z$.
First, suppose that $\mu_i$ enters $T'_z$ before the time when $\mu_j$ leaves $T'_z$.
Since $T'$ is the lowest subtree that has the gap $(\mu_i,\mu_j)$ and w.r.t. \Cref{le:timeLine}, we must have
the gaps $(\mu_i,\mu_p)$ and $(\mu_q,\mu_j)$ on $T'_z$ for some messages $\mu_p \in M^{\uparrow}(T'_z)$ and
$\mu_q \in M^{\downarrow}(T'_z)$.
Using \Cref{le:timeLine} the time when $\mu_p$ reaches $v'_z$ is earlier than the time $\mu_j$
reaches $v'_z$.
Hence, $\mu_p$ reaches $v'$ earlier than the time when $\mu_j$ reaches $v'$ since $\mu_j$ cannot overtake $\mu_p$ using \Cref{le:edgeStates}.
If $\mu_p$ leaves $T'$, then we cannot have the gap $(\mu_i,\mu_j)$ since $\mu_p$ makes the gap $(\mu_i,\mu_p)$ on $T'$.
Consider the case where $\mu_p$ does not leave $T'$. This implies that $\mu_p$ must find a downward link at $v'$
that points to a different child than $v'_z$
when it reaches $v'$ w.r.t. \Cref{le:uniqueEdge} while $v'$ has also an upward link since the time when $\mu_i$ enters $T'$.
Hence, the subtree $T'$ is isolated using \Cref{le:isolatedTree} and $\mu_j$ can never leave $T'$.
Therefore, in either case, the forest $\forest_{\gnn}$ cannot have $(\mu_i,\mu_j)$ as the gap on $T'$ that is a contradiction.

Now, assume that $\mu_i$ enters $T'_z$ after the time when $\mu_j$ leaves $T'_z$. Therefore, $\mu_i$ reaches $v'_z$
after the time when $\mu_j$ reaches $v'_z$.
Since only one message can be in transit on the edge $(v',v'_z)$ w.r.t. \Cref{le:edgeStates}, $\mu_j$
is forwarded to $(v',v'_z)$ by $v'_z$ earlier than the time when $\mu_i$ is forwarded in $(v',v'_z)$ by $v'$.
Therefore $\mu_j$ cannot reach $v'$ later than the time when $\mu_i$ reaches $v'$. On the other hand, it cannot reach $v'$
earlier than the time when $\mu_i$ reaches $v'$ since $(\mu_i,\mu_j)$ is a gap. Hence, they must reach $v'$ at the same time.
This implies that $v'$ must have a downward link using \Cref{le:uniqueEdge} that points to a different child than $v'_z$
and therefore $T'$ must be isolated w.r.t. \Cref{le:isolatedTree} since $v'$ has had also an upward link since $\mu_i$ enters $T'$.
This is a contradiction with the fact that $\mu_j$ must leave $T'$.
Consequently, the claim of the lemma holds.
\end{proof}

\begin{corollary}\label{co:gapLowest}
Consider the lowest subtree $T'$ of $T$ in which $\forest_{\gnn}$ has the gap $(\mu_i,\mu_j)$ on $T'$.
The size of the gap $(\mu_i,\mu_j)$ equals the diameter of $T'$.
\end{corollary}
\begin{proof}
\Cref{le:gapLowest} implies that $v_{\src}(\mu_j)$ and $v_{\des}(\mu_i)$ are in two different children subtrees of $T'$
and therefore the claim of the corollary holds.
\end{proof}

\begin{lemma}\label{le:gapWindow}
Let $T'$ be the lowest subtree of $T$ in which the message $\mu_z$ traverses one of the longest direct paths in $T'$.
Suppose $T'_i$ is the lowest subtree of $T'$ that has the gap $(\mu_i,\mu_z)$ for some message $\mu_i$ that enters $T'_i$.
The forest $\forest_{\gnn}$ has the gap $(\mu_i,\mu_z)$ on all subtrees of $T'$ that are rooted at the nodes
on the direct path between the root of $T'_i$ and either the root of $T'$---excluding the root of $T'$---or the root of
the lowest subtree $T'_j$ of $T'$---excluding the root of $T'_j$---that has the larger gap $(\mu_j,\mu_z)$
for some message $\mu_j$.
\end{lemma}
\begin{proof}
Let $v'$ denote the root of $T'$ and $v'_i$ be the root of $T'_i$.
Further, if there is not a larger gap than $(\mu_i,\mu_z)$, then let $v''$ be $v'$. Otherwise, let $v''$
be the root of $T'_j$. 
The subtree $T'_j$ must have a larger diameter than $T'_i$ w.r.t. \Cref{co:gapLowest}.
Hence, in either case, the subtree rooted at $v''$ has a larger height than $T'_i$.
Assume that $v''_p \neq v''$ is the closest node with $v'_i$ on the direct path between $v''$ and $v'_i$ in which
$(\mu_i,\mu_z)$ is not a gap on the subtree rooted at $v''_p$.
Let $T''_p$ denote the subtree of $T$ that is rooted at $v''_p$.
The message $\mu_z$ must leave $T''_p$ since it visits $v'$ and $T'$ has a larger height than $T''_p$.
Since $\mu_z$ is an actual find-predecessor message, $\mu_z$ must make a gap on $T''_p$ that is larger
than the gap $(\mu_i,\mu_z)$.
This contradicts the fact that $T'_j$ is the lowest subtree that has a larger gap than $(\mu_i,\mu_z)$.
Thus, the claim of the lemma holds.
\end{proof}

\para{Local predecessor:}
Any non-dummy request $r_{\src}(\mu')$ has a unique local predecessor if there is a gap $(\mu,\mu')$
for any message $\mu$ that makes a gap with $\mu'$. The request $r_{\des}(\mu)$ is the local predecessor of $r_{\src}(\mu')$ if $(\mu,\mu')$
be the smallest gap among all gaps $(\mu'',\mu')$ for any message $\mu''$ that makes a gap with the massage $\mu'$.
We may refer to $r_{\des}(\mu)$ as the \textbf{Intra-Component} local predecessor of $r_{\src}(\mu')$
if $(\mu',\mu)$ is an Intra-Component gap. Otherwise, it is called an \textbf{Inter-Component} local predecessor
of $r_{\src}(\mu')$ if $(\mu,\mu')$ is an Inter-Component gap.

For instance, in \Cref{fig:complication}, the request $r^1_0$ is the Intra-Component local predecessor
of $r_b$ and the request $r_b$ is the Inter-Component local predecessor of $r_c$.
We can replace the edge $(r_a,r_b)$ with the new edge $(r^1_0, r_b)$ to satisfy the
Intra-Component property. Further, we can replace the edge $(r^2_0,r_c)$ with the new
edge $(r_b,r_c)$ to satisfy the Inter-Component property.

\begin{remark}\label{re:localPred}
Consider the Intra-Component local predecessor $r_{\des}(\mu)$ of request $r_{\src}(\mu')$. The request $r_{\des}(\mu)$
is indeed scheduled before the request $r_{\src}(\mu')$ in the global schedule. The reason is as follows. Let $T'$ be the lowest subtree of $T$
in which $\forest_{\gnn}$ has the gap $(\mu,\mu')$ on $T'$. Consider the directed path
that is constructed by \gnn\ in the directed graph $H$---the directed version of $T$---during moving $\mu$ from the node of the
successor request of $r_{\des}(\mu)$---that is, $r_{\src}(\mu)$---to the root of $T'$---it continues until $v_{\des}(\mu)$.
This directed path can be deflected towards some successor request of $r_{\src}(\mu)$ in the meantime when $\mu'$ leaves $T'$
until it finds the node of the predecessor request---that is, $r_{\des}(\mu')$. Since $(\mu,\mu')$ is an Intra-Component gap, then $r_{\des}(\mu')$
must be either $r_{\src}(\mu)$ or some request that is scheduled after $r_{\src}(\mu)$.
\end{remark}

\para{Transformation:} Now, we are ready to describe our transformation on $\forest_{\gnn}$ using some modifications.
In each modification, for every request $r$ that has a local predecessor,
we replace the edge in $\forest_{\gnn}$ whose endpoints are $r$ and its actual predecessor with
an edge in ${\R \choose 2}$ whose endpoints are $r$ and its local predecessor.

\begin{figure}[H]
  \center	
  \includegraphics[width=0.2\textwidth]{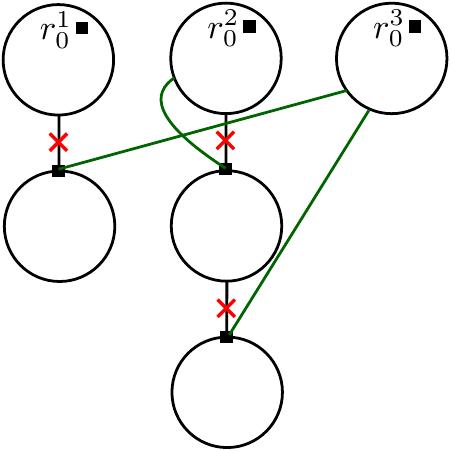}
   \caption{Transforming $\forest_{\gnn}$: The requests are shown with squares.
  There are initially $3$ servers and therefore $3$ dummy requests.
  Hence, the resulted forest by \gnn\ consists of $3$ TSP paths that are denoted in the figure by circles and
  black edges. Each circle is a sub-component of the corresponding TSP path. The black edges are
  replaced with the green edges during the transformation.}
  \label{fig:modification}
\end{figure}

\para{Gap closing:}
A gap $(\mu_i,\mu_j)$ is closed when the edge $e(\mu_j)$ is replaced with the new edge $\big(r_{\des}(\mu_i),r_{\src}(\mu_j)\big)$
in a modification. 

\subsubsection{The Locality-Based Structure of the New Resulted Forest}
\label{sec:mdfGreedy}

In this section, we show that $\forest_{\gnn}$ is transformed into an $\R_D$-respecting spanning $k$-forest of $B$ that satisfies the Property
\ref{pr:ForestIntracomponent} and Property \ref{pr:ForestIntercomponent} after all modifications.
We recall that $\forest_{mdf}$ is the forest that is resulted after the transformation of $\forest_{\gnn}$ that consists of
all possible replacements.

\para{Constructing $\forest_{mdf}$ in a button-up approach:}
We consider all subtrees of the HST $T$ and close the gaps in a button-up approach on each subtree.
We consider $h+1$ steps where $h$ is the height of $T$.
The leaves of $T$ are supposed to be subtrees of $T$ with height $0$.
Assume that the gaps on all subtrees with height less than $i$ are already closed in steps $1,\ldots,i$ for any $i \geq 0$.
Now, we consider any subtree $T'$ of $T$ with height $i$ in the $(i+1)$-th step and we close all gaps on $T'$.
Let $\forest^{i+1}_{mdf}$ denote the result of the transformation at the end of $(i+1)$-th step and
therefore, $\forest_{mdf}=\forest^{h+1}_{mdf}$. Let $\forest^{0}_{mdf}:=\forest_{\gnn}$.

\begin{remark}\label{re:gapclosingLS}
This button-up approach for constructing $\forest_{mdf}$ implies that when a gap is closed in the $(i+1)$-th step on $T'$,
$T'$ is the lowest subtree of $T$ in which $\forest_{\gnn}$ has the gap on $T'$ as otherwise the gap must have been
closed in the previous steps. 
\end{remark}

\begin{corollary}\label{co:gapclosingNewEdge}
Using \Cref{re:gapclosingLS} and \Cref{co:gapLowest}, a new edge that is added as a result of closing any
gap on $T'$ in the $(i+1)$-th step, the new edge connects two requests in two different children subtrees
of $T'$.
\end{corollary}

\begin{corollary}\label{co:gapclosingOriginalEdges}
\Cref{co:gapclosingNewEdge} implies that closing the gaps on $T'$ in the $(i+1)$-th step does not change the structure of $\forest^{i}_{mdf}(T'')$
for any subtree $T''$ of $T'$. Formally, $\forest^{i+1}_{mdf}(T'')=\forest^{i}_{mdf}(T'')$ where $T''$ has height less than $i$.
\end{corollary}

\begin{remark}\label{re:gapclosingUPMsg}
As a result of closing all gaps on $T'$, the edge $e\big(\mu^{\uparrow}_p(T')\big)$ is removed
for any message $\mu^{\uparrow}_p(T') \in M^{\uparrow}(T')$ where $p > 1$.
If $T'$ does not have any dummy request, then $\mu^{\uparrow}_1(T')$ must be an actual message.
However, since $\mu^{\uparrow}_1(T')$ cannot make a gap with any other message in $M^{\downarrow}(T')$ w.r.t. \Cref{le:timeLine}
and the definition of a gap, then $e\big(\mu^{\uparrow}_1(T')\big)$ is not removed.
\end{remark}

\begin{lemma}\label{le:transformationInduction}
Consider any subtree $T'$ of $T$ with height $i$ for any $i \geq 0$ such that $\forest_{\gnn}(T')$ is not empty.
Let $v$ denote the root of $T'$.
Assume that $T'$ includes $m\geq0$ dummy requests. We claim that
\begin{enumerate}[1.]
\item \label{cl:transformationUnchanged} $\forest^{i+1}_{mdf}(T'')=\forest^{i}_{mdf}(T'')$ for any subtree $T''$ of $T'$ that has a smaller height than $T'$.
\item \label{cl:transformationWithDummy} $\forest^{i+1}_{mdf}(T')$ has $\max\{1,m\}$ components in which any of these $m$ components is a connected
tree that includes at most one dummy request.
\end{enumerate}
\end{lemma}
\begin{proof}
The Claim \ref{cl:transformationUnchanged} is true w.r.t. \Cref{co:gapclosingOriginalEdges} for all $i \geq 0$.
Let $i=0$ and hence consider any subtree $T'$ of $T$ with height $0$.
The subtree $T'$ is actually a leaf node.
The subtree $T'$ either does not host any dummy request or it hosts exactly one dummy request w.r.t. the \dsms\ problem definition.
From a theoretical point of view, a leaf node can invoke more than one request at the same time in a one-shot execution.
However, all requests are scheduled consecutively in a one-shot execution w.r.t. the description of \gnn.
Hence, $T'$ does not have any gap since there is only one message in $M^{\uparrow}(T')$.
Therefore, all requests on $u$ are connected as a TSP path that is a sub-component of some component of $\forest_{\gnn}$.
This TSP path remains unchanged at the end of the first step since $T'$ does not have any gap.
In fact, $\forest^{1}_{mdf}(T')=\forest_{\gnn}(T')$.
The tail of this TSP path is $r_{\src}\big(\mu^{\uparrow}_1(T')\big)$ that is either a dummy request if $T'$ hosts one dummy request
or is an actual request if $T'$ does not host any dummy request. Consequently, the Claim \ref{cl:transformationWithDummy}
holds for any leaf node. 

Assume that $i\geq1$ steps have been done and the Claim \ref{cl:transformationWithDummy} now holds for any subtree with
height less than $i$. Consider the $(i+1)$-th step and any subtree $T'$ of $T$ with height $i$ that includes $m \geq 0$ dummy requests
such that $\forest_{\gnn}(T')$ is not empty.
Some of the gaps on $T'$ might have been already closed in the previous steps.
We show that the Claim \ref{cl:transformationWithDummy} holds for $T'$ at the end of $(i+1)$-th step since a) any component
of $\forest^{i+1}_{mdf}(T')$ includes at most one dummy request, b) any component of $\forest^{i+1}_{mdf}(T')$ is a tree, and
c) $\forest^{i+1}_{mdf}(T')$ has $\max\{1,m\}$ components. 

\begin{enumerate}[a.]
\item Any component of $\forest^{i+1}_{mdf}(T')$ includes at most one dummy request:
Using \Cref{co:gapclosingOriginalEdges} and w.r.t. our assumption of the induction hypothesis,
any component of $\forest^{i+1}_{mdf}(T'_z)$ for any child subtree $T'_z$ of $T'$ includes
at most one dummy request. Therefore, w.r.t. \Cref{co:gapclosingNewEdge} we need to show that there is not any edge in $\forest^{i+1}_{mdf}(T')$
that connects one component of $\forest^{i+1}_{mdf}(T'_z)$ that includes a dummy request
and another component of $\forest^{i+1}_{mdf}(T'_w)$ that also includes a dummy request
where $T'_z$ and $T'_w$ are children subtrees of $T'$.
The subtrees $T'_z$ and $T'_w$ must include dummy requests.
Since $T'_w$ and $T'_z$ include dummy requests and have height $i-1$, all edges $e\big(\mu^{\uparrow}_p(T'_w)\big)$
and $e\big(\mu^{\uparrow}_p(T'_z)\big)$ for $p \geq 2$ are removed in steps $1,\dots,i$ w.r.t. \Cref{re:gapclosingUPMsg}
and the messages $\mu^{\uparrow}_1(T'_w)$ and $\mu^{\uparrow}_1(T'_z)$ are both virtual messages w.r.t. the definition
of $M^{\uparrow}(T'')$ for any subtree $T''$ of $T$. This implies that: First, any original edge---an edge in $\forest_{\gnn}$---that connects
two requests in $T'_w$ and $T'_z$ has been removed in previous steps. Second, there is not any open gap on $T'$ that is
made with a message in $M^{\uparrow}(T'_z)$ or in $M^{\uparrow}(T'_w)$. Thus, closing gaps on $T'$ in the $(i+1)$-th step
cannot add a new edge between $T'_w$ and $T'_z$.

\item Any component of $\forest^{i+1}_{mdf}(T')$ is a tree:
Using \Cref{co:gapclosingOriginalEdges} and w.r.t. our assumption of the induction hypothesis,
any component of $\forest^{i+1}_{mdf}(T'_z)$ for any child subtree $T'_z$ of $T'$ is a tree.
Therefore, w.r.t. \Cref{co:gapclosingNewEdge} we need to show that there is not any cycle that consists of some edges
in $\forest^{i+1}_{mdf}(T')$ such that every edge of the cycle connects two children subtrees of $T'$. For the sake of
contradiction, assume that there is such a cycle. We see the ``nodes'' of the cycle as a subset of children subtrees of $T'$.
Any child subtree of $T'$ that includes a dummy request, cannot be a node of the cycle since all edges corresponding with
the message that leave the subtree with height $i-1$ have been removed in previous steps w.r.t. \Cref{re:gapclosingUPMsg}.
Hence, any node of the cycle must be a child subtree of $T'$ that does not include any dummy request.
Let $J$ denote the set of those children subtrees of $T'$ that are seen as the nodes of the cycle.
Consider the first find-predecessor message $\mu$ that is processed by $v$---that is, the root of $T$---among all messages
that leaves the children subtrees in $J$.
The message $\mu$ either leaves $T'$ or finds a downward link on $v$.
In the latter case, the downward link does not point to any request that is in a subtree in $J$ since $\mu^{\uparrow}_1(T'_z)$
is an actual message for any $T'_z$ in $J$ and $\mu$ is the first message among all messages in $\cup_{\forall T'_z \in J} M^{\uparrow}(T'_z)$
that is processed by $v$, and therefore a message cannot enter any subtree in $J$ before $\mu$ reaches $v$ \Cref{le:timeLine}.
Therefore, the request $r_{\des}(\mu)$ is in a child subtree of $T'$ that is not in $J$ and consequently
the cycle is broken that is a contradiction.
Consider the former case where $\mu$ leaves $T'$.
Hence, using \Cref{le:timeLine} and w.r.t. the definition of a gap the message $\mu$ must make a gap $(\mu',\mu)$ on $T'$
for some message $\mu'$ that enters $T'$.
The gap $(\mu',\mu)$ is closed in $(i+1)$-th step.
Since $\mu$ is the first message among all messages in
$\cup_{\forall T'_z \in J} M^{\uparrow}(T_z')$ that reaches the root of $T'$ and processed by the root of $T'$, and
the time when $\mu$ reaches the root of $T'$ is not earlier than the time $\mu'$ reaches the root of $T'$ because
$(\mu',\mu)$ is a gap on $T'$, then $\mu'$ cannot enter any subtree in $J$.
As a result of closing the gap $(\mu',\mu)$ a new edge is added w.r.t. our transformation described in \Cref{sec:modification}.
The new edge must connect two requests $r_{\src}(\mu)$ and $r_{\des}(\mu')$ in two different children subtrees of $T'$
using \Cref{co:gapclosingNewEdge}.
Since $\mu'$ does not enter any subtree in $J$, the new edge must connect a subtree in $J$ with a child subtree of $T'$
that is not in $J$.
We again get a contradiction since the cycle is broken. Consequently, there is not such a cycle and any component
of $\forest^{i+1}_{mdf}(T')$ is indeed a tree.

\item $\forest^{i+1}_{mdf}(T')$ has exactly $\max\{1,m\}$ components:
If $\forest_{\gnn}(T')$ is not empty, $\forest^{i+1}_{mdf}(T')$ cannot have less than $\max\{1,m\}$ components
w.r.t. the fact---it is already proved---that any component of $\forest^{i+1}_{mdf}(T')$ includes at most one dummy request.
By contrast, assume that $\forest^{i+1}_{mdf}(T')$ has more than $\max\{1,m\}$ components for the sake of contradiction.
Therefore, there are at least two components in $\forest^{i+1}_{mdf}(T')$ and at least one of the components in $\forest^{i+1}_{mdf}(T')$
does not include any dummy request since $T'$ includes $m$ dummy requests. Consider any component $F^z$ of $\forest^{i+1}_{mdf}(T')$
that does not include any dummy request. The first message among all messages corresponding with the requests
in $F^z$ that is processed by $v$, must leave $T'$ as otherwise the message is forwarded
inside $T'$ towards some request that is not in $F^z$ w.r.t. \Cref{le:timeLine} and therefore $F^z$ must
also include some other request in $T'$. Let $F^z$ be the component of $\forest^{i+1}_{mdf}(T')$ that does not include
any dummy request such that the first message $\mu$ corresponding with any request included in $F^z$ leaves $T'$ after
$\mu^{\uparrow}_1(T')$. The message $\mu$ must make a gap $(\mu',\mu)$ on $T'$ w.r.t. \Cref{le:timeLine} and the definition
of a gap for some message $\mu'$ that enters $T'$. However, the gap $(\mu',\mu)$ is closed in the $(i+1)$-th step
and therefore the local predecessor $r_{\des}(\mu')$ must be in $F^z$. This is a contradiction, since $\mu$ is the first message
that leaves $T'$ among all messages corresponding with requests in $F^z$ and therefore $r_{\des}(\mu')$ cannot be in $F^z$
\Cref{le:timeLine}.
\end{enumerate}
\end{proof}

\begin{proof}[\textbf{Proof of \Cref{le:transformation}}]
The Claim \ref{cl:transformationWithDummy} of \Cref{le:transformationInduction} implies that at the end of $(h+1)$-th step of
our button-up constructing of $\forest_{mdf}$, $\forest^{h+1}_{mdf}(T)=\forest_{mdf}$ is an $\R_D$-respecting spanning $k$-forest of $B$
since the HST $T$ includes $k \geq 1$ dummy requests.

Further, the Property \ref{pr:ForestIntracomponent} and Property \ref{pr:ForestIntercomponent} are guaranteed on any subtree $T'$ of $T$ with height $i \geq 0$
using the Claim \ref{cl:transformationWithDummy} of \Cref{le:transformationInduction} at the end of $(i+1)$-th step. The two properties are not also violated
later in the $j$-th step where $j > i+1$ w.r.t. the Claim \ref{cl:transformationUnchanged} of \Cref{le:transformationInduction}.
\end{proof}

\subsubsection{Total Cost of GNN: An Upper Bound}
\label{sec:gnnUB}

In this section, we show that the total weight of $\forest_{mdf}$ is at least the total cost of $\forest_{\gnn}$. Formally, we want to show that

\begin{equation}\label{eq:GNNLatencyUB}
	\Latency_{\gnn}(\forest_{\gnn}) \leq \Weight_T(\forest_{mdf}).
\end{equation}

The latency of a message $\mu$ from $u$ to $v$ on $T$ is denoted by $\latency_{\gnn}(\mu,u,v)$. Hence,
$\latency_{\gnn}\big(e(\mu)\big):=\latency_{\gnn}(\mu)=\latency_{\gnn}(\mu,v_{\src}(\mu),v_{\des}(\mu))$.
Let $(\mu,\mu')$ be the smallest gap among all gaps $(\mu'',\mu')$ for any message $\mu''$ that makes a gap with $\mu'$.
Therefore, $e(\mu')$ is removed and replaced with a new edge that connects two requests $r_{\des}(\mu)$ and $r_{\src}(\mu')$.
Since we want to show that \eqref{eq:GNNLatencyUB} holds, we need to upper bound the latency of $\mu'$---that is, $\latency_{\gnn}(\mu')$---with
the weight of the new edge $e^{new}=\big(r_{\des}(\mu),r_{\src}(\mu')\big)$. However, the $\latency_{\gnn}(\mu')$ can be larger
than $\weight_T(e^{new})$. By contrast, the following lemma shows that the latency of $\mu$ is upper bounded by $\weight_T(e^{new})$.
This lemma gives us the go-ahead to show that $\weight_T(e^{new})$ can be seen as an ``amortized'' upper bound for
$\latency_{\gnn}(\mu')$. 

\begin{lemma}\label{le:greedyNature}
Consider any subtree $T'$ of $T$. Assume that $\forest_{\gnn}$ has any gap $(\mu,\mu')$ on $T'$.
The latency of the message $\mu$ is upper bounded by the diameter of $T'$. Formally, 
\[
	\latency_{\gnn}(\mu) \leq \delta(T').
\] 
\end{lemma}
\begin{proof}
Let $v$ denote the root of $T'$.
Using \Cref{le:uniquePath}, a find-predecessor message always
finds the node of its predecessor using a direct path constructed by \gnn. Using \Cref{le:timeLine},
the message $\mu'$ reaches $v$ not earlier than the time when $\mu$ reaches $v$. Therefore, 
\begin{equation}\label{eq:localSucc}
	\latency_{\gnn}\big(\mu,v_{\src}(\mu),v\big) \leq \latency_{\gnn}\big(\mu',v_{\src}(\mu'),v\big).
\end{equation}
On the other hand we can have:
  \begin{eqnarray*}
    \latency_{\gnn}(\mu)
    & = &
          \latency_{\gnn}\big(\mu,v_{\src}(\mu),v\big) + \Big(\latency_{\gnn}(\mu) -
          \latency_{\gnn}\big(\mu,v_{\src}(\mu),v\big)\Big)\\
    & \stackrel{\eqref{eq:localSucc}}{\leq} &
             \latency_{\gnn}\big(\mu',v_{\src}(\mu'),v\big) + \Big(\latency_{\gnn}(\mu) -
          \latency_{\gnn}\big(\mu,v_{\src}(\mu),v\big)\Big)\\
    & \leq &
             d_T\big(v_{\src}(\mu'),v\big) + d_T\big(v,v_{\des}(\mu)\big)\\
    & = &
    	  \delta(T').
  \end{eqnarray*}
  The second inequality follows because the latency of an edge is at most the weight of the edge (see \Cref{sec:communicationModel}).
\end{proof}

\para{Amortized analysis:}
In the overview of our amortization that has been provided in \Cref{sec:gnnHST}
using the simple case, we use the potential $\weight_T(e^{pot}_i) - \latency_{\gnn}(e^{pot}_i)$
to take $\weight_T(e^{new}_i)$ into account as an amortized upper bound of $\latency_{\gnn}\big(\mu(e^{old}_i)\big)$. In general, the same approach can be used
to show that \eqref{eq:GNNLatencyUB} holds with a more complicated analysis. The complication appears in two directions: 1) when $|E^{old}(e)| >1$
for any edge $e \in E^{pot}$ or $|E^{pot}(e)| >1$ for any edge $e \in E^{old}$. Further, 2) the sets $E^{old}$ and $E^{pot}$ can share some edges that create
a \textbf{dependency graph} between the edges in $E^{old}$ and $E^{pot}$. Because of the latter situation, we cannot replace an edge
$e \in E^{pot} \cap E^{old}$ before replacing the edges in $E^{old}(e)$ since we require the potential of the edge $e$ for amortizing
the edges in $E^{old}(e)$.
Hence, we consider priorities for edges to be replaced. Regarding to the first complication, we will show that for any edge $e \in E^{old}$,
all edges in $E^{pot}(e)$ contribute enough potential for amortizing the weight of $e$. Further, we will guarantee that the maximum potential of any edge
$e \in E^{opt}$ that is distributed among all edges in $E^{old}(e)$ is $\weight_T(e) - \latency_{\gnn}(e)$.

\para{Priority directed graph (PDG):}
We consider the dependencies between the edges in $E^{old}$ and $E^{pot}$ as a priority directed graph in which any edge in $E^{old} \cup E^{pot}$
is represented by a node in PDG and there is a directed edge from the node $e$ to another node $e'$ in PDG if $e' \in E^{old}$ and $e \in E^{pot}(e')$.   

\begin{lemma}\label{le:PDG}
The priority directed graph (PDG) is acyclic.
\end{lemma}
\begin{proof}
For the sake of contradiction, assume that there is a cycle in PDG consists of $e_1, e_2, \dots,e_i$ such that $e_p$ points to $e_{p+1}$ where $p \in [1,i-1]$
and $e_i$ points to $e_1$. When an edge $e$ points to another edge $e'$ in PDG, it implies that there is a gap $\big(\mu(e),\mu(e')\big)$ in $\forest_{\gnn}$.
Let $T_p$ be any subtree of $T$ in which $\forest_{\gnn}$ has the gap $\big(\mu(e_{p-1}),\mu(e_p)\big)$ on it for $p \in [2,i]$ and let $T_1$
be any subtree of $T$ such that $\forest_{\gnn}$ has the gap $\big(\mu(e_i),\mu(e_1)\big)$ on it. Let $t^{\uparrow}_p$ denote the time
when $\mu(e_p)$ leaves $T_p$ for all $p \in [1,i]$ and $t^{\downarrow}_p$ denote the time when it enters $T_{p+1}$ for all $p \in [1,i-1]$.
Finally, let $t^{\downarrow}_i$ denote the time when it enters $T_1$. Using the definition of a gap and \Cref{le:timeLine}, we have

\[
	\forall p \in [2,i] : t^{\downarrow}_{p-1} \leq t^{\uparrow}_p
\]
and 
\[
	t^{\downarrow}_i \leq t^{\uparrow}_1.
\]

On the other hand, we have $t^{\uparrow}_p < t^{\downarrow}_p$ for all $p \in [1,i]$. Therefore, we get $t^{\downarrow}_p < t^{\downarrow}_p$
for all $p \in [1,i]$ that is a contradiction. Consequently, PDG is acyclic. 
\end{proof}

For our analysis, we consider steps, and in each step, we replace those edges in $E^{old}$ whose corresponding nodes in PDG do not have
any outgoing edges. Note that in each step until the end when all edges in $E^{old}$ are replaced, there is such a node in PDG w.r.t. \Cref{le:PDG}.
Further, we remove these nodes and the incident edges from PDG at the end of each step.
 
The following lemma shows that the total weight of the edges that are removed during the transformation of $\forest_{\gnn}$ are amortized with the
total weight of edges that are added during the transformation.

\begin{lemma}\label{le:amortization}
In the overview of our amortized analysis, using \Cref{le:amortizedSimple}, it was shown that \eqref{eq:cumulativeAmortization} holds
for the simple case. Even in a general case \eqref{eq:cumulativeAmortization} holds.
\end{lemma}
\begin{proof}
Consider the edge $e^{old}_z \in E^{old}$ that does not have any outgoing edge in PDG in the current step of transformation.
Suppose that during the transformation, $e^{old}_z$ is replaced with $e^{new}_z$.
Let $T'$ be the subtree of $T$ such that $\weight_T(e^{old}_z)=\delta(T')$.
Further, let $E^{pot}(e^{old}_z)=\big\{e^{pot}_1,e^{pot}_2,\dots,e^{pot}_i\big\}$.
We assume that, w.l.o.g., the gap $\big(\mu(e^{pot}_{p+1}),\mu(e^{old}_z)\big)$ is larger than the gap $(\mu(e^{pot}_p),\mu(e^{old}_z))$
for all $p \in [1,i-1]$.
Suppose $T'_p$ is the lowest subtree of $T'$ such that $\forest_{\gnn}$ has the gap $(\mu^{pot}_p,\mu^{old}_z)$ on that for all $p \in [1,i]$. 
Let $v_p$ denote the root of $T'_p$ and $v$ denote the root of $T'$. 
For simplicity, let $\mu^{old}_q=\mu(e^{old}_q)$, $\mu^{new}_q=\mu(e^{new}_q)$, and $\mu^{pot}_q=\mu(e^{pot}_q)$ for all $q$.
Using \Cref{co:gapLowest}, $T'_{p+1}$ must be higher than $T'_p$ as you can see in \Cref{fig:edgeAmortized}. 
\begin{figure}[H]
  \center	
    \includegraphics[width=0.5\textwidth]{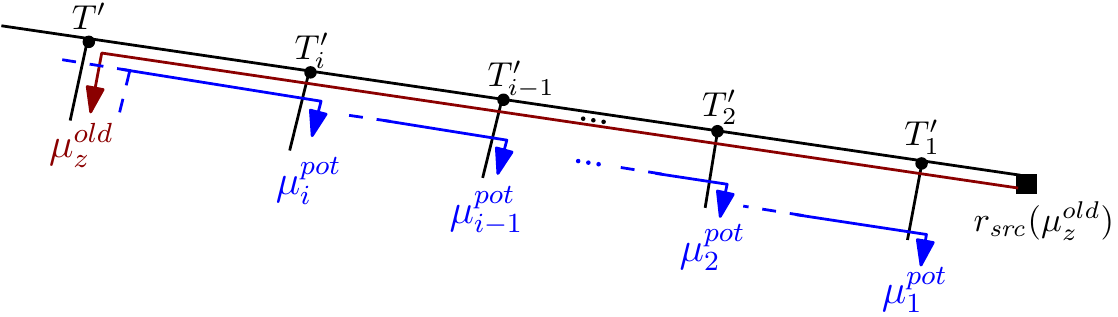}
  \caption{All messages that make gaps with $\mu^{old}_z$.}
  \label{fig:edgeAmortized}
\end{figure}
The gap $(\mu^{pot}_1,\mu^{old}_z)$ is the smallest gap among all gaps $(\mu^{pot}_p,\mu^{old}_z)$ for any $p \in [1,i]$
using the assumption of the lemma. Since $T'_1$ is the lowest subtree of $T$ that has the gap $(\mu^{pot}_1,\mu^{old}_z)$,
w.r.t. \Cref{co:gapLowest}, therefore, 
\[
	\weight_T(e^{new}_z)=\delta(T'_1)=2 \cdot d_T\left(v_{\src}(\mu^{old}_z),v_1\right).
\]
Since $\weight_T(e^{old}_z)=2\cdot d_T\big(v_{\src}(\mu^{old}_z),v\big)$, hence
we have
\begin{equation}\label{eq:edgeUBGeneralCase}
	\weight_T(e^{old}_z) = \weight_T(e^{new}_z) + 2 \cdot \left(d_T(v_i,v) + \sum_{p=1}^{i-1} d_T(v_p,v_{p+1}) \right).
\end{equation}
\Cref{le:gapWindow} implies that all subtrees rooted at the nodes on the direct path between $v_p$ and $v_{p+1}$ for any $p \in [1,i-1]$
cannot have another gap $\big(\mu^{pot}_p,\mu'\big)$ for any edge $e(\mu') \in E^{old}$ and all subtrees rooted at the nodes on the direct path
between $v_i$ and $v$ cannot have another gap $\big(\mu^{pot}_i,\mu'\big)$ for any edge $e(\mu') \in E^{old}$.
Therefore, the direct path between $v_p$ and $v_{p+1}$ for any $p \in [1,i-1]$ is a sub-path
of the whole path that is traversed by $\mu^{pot}_p$. Further, the direct path between
$v_i$ and $v$ is a sub-path of the whole path that is traversed by $\mu^{pot}_i$.
Consequently, the amount $d_T(v_p,v_{p+1})$ for any $p \in [1,i-1]$ as well as the amount $d_T(v_i,v)$
are counted exactly once and do not appear in the right side of \eqref{eq:edgeUBGeneralCase} for any other edge in
$E^{old}$. Therefore we can have the following.
Consider any edge $e$ in $E^{pot}$. Assume that we sum up the right side of \eqref{eq:edgeUBGeneralCase} for all $e^{old}_z \in E^{old}$.
Regarding to the edge $e$, at most the amount $\weight_T\big(e)-\delta(T_1(e)\big)$ appears
where $\delta\big(T_1(e)\big)$ is the diameter of the lowest subtree of $T$ in which $\forest_{\gnn}$ has the gap $(\mu(e),\mu')$ for any $e(\mu') \in E^{old}$.
Note that the amount $2 \cdot d_T\big(v_{\des}(\mu(e)),v_1(e)\big)$ where $v_1(e)$ is the root of $T_1(e)$ is never used in the right side of
\eqref{eq:edgeUBGeneralCase} for any edge in $E^{old}$ and therefore $\delta\big(T_1(e)\big)$ can be subtracted when we sum up the right side of
\eqref{eq:edgeUBGeneralCase} for all $e^{old}_z \in E^{old}$.
Hence, we sum up \eqref{eq:edgeUBGeneralCase} for all $e^{old}_z \in E^{old}$ and we get
\begin{equation}\label{eq:subtractedAmortization}
	\Weight_T(E^{old}) \leq \Weight_T(E^{new}) + \left(\Weight_T(E^{pot}) - \sum_{e \in E^{pot}} \delta(T_1(e)) \right).	
\end{equation}
On the other hand, using \Cref{le:greedyNature} we can have the following for every $e \in E^{pot}$,
\[
	\latency_{\gnn}(e) \leq \delta\big(T_1(e)\big).
\]
If we sum up the above equation for all edges in $E^{pot}$, we get
\begin{equation}\label{eq:sumEpot}
	\Latency_{\gnn}(E^{pot}) \leq \sum_{e \in E^{pot}} \delta\big(T_1(e)\big).
\end{equation}
Using \eqref{eq:subtractedAmortization} and \eqref{eq:sumEpot} we get
\begin{equation}\label{eq:cumulativeAmortizationGeneral}
		\Weight_T(E^{old}) \leq \Weight_T(E^{new}) + \left(\Weight_T(E^{pot}) - \Latency_{\gnn}(E^{pot}) \right).	
\end{equation}
Regarding to the definition of the potential function $\Phi$, the claim of the lemma holds.
\end{proof}

\begin{proof}[\textbf{Proof of \Cref{le:UBGNNForest}}]
Although for our analysis we replace the edges w.r.t. the priority directed graph (PDG), we still use the potential
of edges in $E^{old} \cap E^{pot}$ for computing \eqref{eq:cumulativeAmortizationGeneral} while they are removed
during the transformation. However, we show that
$\Phi\left(E^{pot} \cap E^{old}\right)$ is actually used only as an auxiliary potential and can be removed. 

Consider \Cref{le:amortization}. If we subtract $\Weight_T(E^{pot} \cap E^{old})$ from both side of \eqref{eq:cumulativeAmortizationGeneral}, then we get
\[
		\Weight_T\left(E^{old} \setminus E^{pot}\right) \leq
		\Weight_T\left(E^{new}\right) + \Weight_T\left(E^{pot} \setminus E^{old}\right) - \Latency_{\gnn}(E^{pot}).	
\]
Since w.r.t \eqref{eq:latencyUB} $\Latency_{\gnn}\left(E^{old} \setminus E^{pot}\right) \leq \Weight_T\left(E^{old} \setminus E^{pot}\right)$
and $E^{old}=\left(E^{old} \setminus E^{pot}\right) \cup \left(E^{old} \cap E^{pot} \right)$ we can have
\[
		\Latency_{\gnn}(E^{old}) \leq
		\Weight_T\left(E^{new}\right) + \Weight_T\left(E^{pot} \setminus E^{old}\right) - \Latency_{\gnn}\left(E^{pot} \setminus E^{old}\right).	
\]
Thus, the actual potential that is used for amortizing $\Latency_{\gnn}(E^{old})$ is $\Phi(E^{pot} \setminus E^{old})$
and the edges in $E^{pot} \setminus E^{old}$ are indeed in $\forest_{mdf}$.
Since we have $E^{old} \cup E^{pot}=\left(E^{old}  \cup \left(E^{pot} \setminus E^{old} \right)\right)$, therefore
\[
		\Latency_{\gnn}\left(E^{old} \cup E^{pot}\right) \leq
		\Weight_T\left(E^{new}\right) + \Weight_T\left(E^{pot} \setminus E^{old}\right).	
\]
Using \eqref{eq:latencyUB} we can have $\Latency_{\gnn}\left(\forest_{\gnn} \setminus \left(E^{old} \cup E^{pot} \right) \right) \leq
\Weight_T\left(\forest_{\gnn} \setminus \left(E^{old} \cup E^{pot} \right) \right)$. Together with $\forest_{mdf}=\forest_{\gnn} \setminus E^{old} \cup E^{new}$
we get
\[
	\Latency_{\gnn}(\forest_{\gnn}) \leq \Weight_T(\forest_{mdf}).
\]
\end{proof}

\begin{proof}[\textbf{Proof of \Cref{th:gnn}}]
The claim of the theorem immediately follows \Cref{le:transformation} and \Cref{le:UBGNNForest}.
\end{proof}

\begin{proof}[\textbf{Proof of \Cref{th:hstOptimal}}]
\Cref{th:genericResult} and \Cref{th:gnn} both together show that the claim of the theorem holds.
\end{proof}

\begin{proof}[\textbf{Proof of \Cref{co:hstOptimalQueuing}}]
As described in \Cref{sec:intro}, the distributed queuing problem is an application of
\dsms\ where $k=1$. The goal in the distributed queuing problem is to minimize sum
of the total communication cost and the total waiting time. When all requests are simultaneously invoked,
the total waiting time gets $0$. Consequently, \gnn\ optimally solves the distributed queuing problem
for one-shot executions on HSTs in the light of \Cref{th:hstOptimal}.
\end{proof}

\subsection{\dsms\ Problem on General Networks}
\label{sec:DSMSonGeneralNetworks}

In this section, we consider a general graph $G=(V,E)$ as the input graph. We show that when running any distributed \dsms\ protocol
\alg\ that satisfies the conditions of \Cref{th:genericResult} on top of HST $T$, we obtain a randomized protocol
with an expected competitive ratio of at most $\bigO{\log n}$ against an oblivious adversary.
The following theorem provides a general version of \Cref{th:graph}.

\begin{theorem}\label{th:genericGraph}
Suppose we are given a graph $G = (V,E)$ and a set of requests $\R$ that all are invoked at the same time by the nodes of $G$
where $|V|=n$. There is a randomized embedding of $G$ into a distribution over HSTs in which we sample an HST $T$ according to the
distribution defined by the embedding. Consider any distributed protocol \alg\ that that satisfies the conditions of \Cref{th:genericResult}.
When running \alg\ on $T$, we get a distributed randomized protocol for $G$ with an expected competitive ratio of at most $\bigO{\log n}$
against an oblivious adversary. This even holds if communication is asynchronous.
\end{theorem}
\begin{proof}
Assume that the HST $T$ is constructed on top of $G$ by using the randomized algorithm of \cite{fakcharoenphol2003tight}.
Let $\forest_{\opt}(T)$ and $\forest_{\opt}(G)$ denote the resulted forests by \opt\ where communication is synchronous on $T$ and $G$ as the input graph.
We get

\begin{equation}\label{eq:optUBGraph}
    \E\left[\Weight_{T}\left(\forest_{\opt}(T)\right)\right] \leq
    \E\left[\Weight_{T}\left(\forest_{\opt}(G)\right)\right] \leq
    \bigO{\log n} \cdot \Weight_{G}\left(\forest_{\opt}(G)\right).
\end{equation}

The first inequality follows from the fact that $\forest_{\opt}(G)$ is not necessarily an
optimal weight forest w.r.t. the edge weights of $T$. The second inequality follows from the expected stretch bound of the HST
construction of \cite{fakcharoenphol2003tight}. Given \Cref{th:genericResult} and \eqref{eq:optUBGraph}, the claim of the theorem holds w.r.t. \Cref{re:asynchOpt}.
In fact, using the assumptions of \Cref{th:genericResult} that is also stated in the theorem, \Cref{le:optimalGreedyForest}, and \eqref{eq:optForestCostLB}
we have $\Delay_{\alg} \leq \Weight_{T}\left(\forest_{\opt}(T)\right)$.
\end{proof}

Note that the statement of \Cref{th:genericGraph} also holds when communication is synchronous on $G$. 
Because the statement of \Cref{th:genericResult} applies to the general asynchronous case, it also captures a synchronous scenario,
where the latency on each edge is fixed but might be smaller than the actual weight of the edge on $T$.
Note that such executions are relevant because an HST is often built as an overlay graph on top of an underlying network
graph $G$ and the latency of simulating a single HST edge might be smaller than the weight of the edge.

\begin{proof}[\textbf{Proof of \Cref{th:graph}}]
Given \Cref{th:gnn} and \Cref{th:genericGraph}, we immediately get the theorem.
\end{proof}
\section{Lower Bound}
\label{sec:LB} 
We provide a simple reduction from the distributed $k$-server problem \cite{bartal1992distributed}
(a statement of this problem is given in \Cref{sec:intro}) to the \dsms\ problem that preserves the competitive ratio up to some constant factor.
We utilize the lower bound presented in \cite{bartal1992distributed} together with our reduction to prove our
lower bound stated in \Cref{thm:LB}.

The reduction is trivial since the \dsms\ problem when requests are sequentially invoked
is identical with the distributed $k$-server problem but their cost functions.
We consider instances that consist of a synchronous network that is modeled by a graph $G=(V,E)$
and a set of requests that are sequentially invoked one by one. Let $I$ denote such an instance.
Our reduction neither changes
the sequence of request nor the input network.
The only difference is with respect to their cost functions. Hence, let us provide our analysis of their costs.

\begin{proof}[\textbf{Proof of \Cref{thm:LB}}] As before, let $\R$ denote a sequence of requests including the $k$ dummy requests in $\R_D$.
Let $\forest$ denote the resulted forest by an optimal offline protocol that solves the instance $I$. $\forest$ consists of $k$
TSP paths that span all requests in $\R$.
The total communication cost incurred by the optimal offline \dsms\ protocol equals the
total weight of the forest $\forest$, that is, $\Weight_G(\forest)$ (see \eqref{eq:forestAlgLength}) since communication is synchronous.
Suppose that there is a $c$-competitive online \dsms\ protocol that solves the instance $I$.
The cost incurred by the online \dsms\ protocol is at most $c \cdot \Weight_G(\forest)$.
However, the total cost incurred by the online distributed $k$-server protocol generated by our reduction on the set of
requests $\R$ is $\bigO{c\cdot D \cdot \Weight_G(\forest)}$ where $D$ is the ratio between the cost
to move a server and the cost to transmit a message over the same distance in synchronous networks.
Note that $c \cdot \Weight_G(\forest)$ is the maximum total weight of the resulted forest by the online protocol.
Therefore, $c \cdot D \cdot \Weight_G(\forest)$ is the maximum total movement cost of all servers of the online protocol.
On the other hand, the total cost incurred by an optimal offline protocol for the distributed $k$-server problem
is $\Omega\big(D \cdot \Weight_G(\forest)\big)$. Consequently, the claim of the theorem holds using also the lower bound of $k$ \cite{manasse1988competitive}.
\end{proof}

\bibliographystyle{alpha}
\bibliography{references} 

\appendix

\section{Minimum Spanning Forest Approximation}
\label{appsec:msf} 

In the following, we prove a generic result about spanning forests of a weighted graph
$G=(V,E,w)$, and let \forest\ denote a forest that spans the nodes in $V$. For any edge
$e\in \forest$, let $w(e)$ denote the weight of the edge on $G$. Further, assume that
when removing $e$ from \forest, the node sets of the resulting $(m+1)$ ($m > 0$)
connected components are $V_{e^1},\ldots,V_{e,m+1}$. Let $(V_{e^1},\ldots,V_{e,m+1})$
be the $(m+1)$-cut induced by removing $e$ from \forest. If $m=1$,
\Cref{appthm:MSFapprox} in particular implies the following result about a spanning tree $T$
of a weighted graph $G$. If for every edge $e\in T$, and every edge $e^*$ over the cut induced by
$T$ when removing $e$ from $T$ it holds that $w(e^*)\geq w(e)/\lambda$, then the total weight
of $T$ is within a factor $\lambda$ of the total weight of a minimum spanning tree (MST) of $G$.

\begin{theorem}\label{appthm:MSFapprox}
  Let $\lambda \geq 1$, $m \geq 1$, and $G=(V,E,w)$ be a weighted connected
  graph with non-negative edge weights $w(e)\geq 0$. Further, let
  $S\subseteq V$, $|S|\leq m$ and let \forest\ and $\forest^*$ be two
  arbitrary $S$-respecting spanning $m$-forests of $G$. Further assume
  that for every pair $(e,e^*)$ of edges $e\in \forest$ and $e^*\in \forest^*$ such that
  $\forest \setminus\set{e}\cup\set{e^*}$ is an $S$-respecting spanning
  $m$-forest of $G$, it holds that $w(e^*)\geq w(e)/\lambda$. Then,
  the total weight of all edges of $\forest$ is at most $\lambda$
  times the total weight of the edges of $\forest^*$.
\end{theorem}

\begin{proof}
  For an edge set $F\subseteq E$, we use $W(F)$ to denote the total
  weight of the edges in $F$. We prove the stronger statement that
  \begin{equation}\label{appeq:MSFapprox}
    W(\forest\setminus \forest^*) \leq \lambda\cdot 
    W(\forest^*\setminus \forest).
  \end{equation}
  We show \eqref{appeq:MSFapprox} by induction on
  $|\forest\setminus \forest^*| = |\forest^*\setminus \forest|$. First
  note that if $|\forest\setminus \forest^*| = 0$, we have
  $\forest=\forest^*$ and thus \eqref{appeq:MSFapprox} is clearly
  true. Further, if $|\forest\setminus \forest^*|=1$, there is exactly
  one edge $e\in \forest\setminus \forest^*$ and exactly one edge
  $e^*\in \forest^*\setminus \forest$. We therefore have
  $\forest^*=\forest\setminus\set{e}\cup\set{e^*}$ and by the
  assumptions of the theorem we have $w(e)\leq \lambda \cdot w(e^*)$,
  implying \eqref{appeq:MSFapprox}.

  Let us therefore assume that
  $|\forest\setminus \forest^*|=\gamma\geq 2$ and let $e$ be a maximum
  weight edge of $\forest\setminus \forest^*$. Let
  $(V_{e^1},\ldots,V_{e,m+1})$ be the $(m+1)$-cut induced by removing
  $e$ from $\forest$. Let $\forest'$ be a
  spanning forest of $G$ that is obtained by removing $e$ from
  $\forest$ and by adding some edge
  $e^* \in \forest^*\setminus \forest$ that connects two components
  $V_{e,i}$ and $V_{e,j}$ where $1 \leq i \neq j \leq m+1$ such that
  $\forest'$ is an $S$-respecting $m$-forest.  Note that such an edge
  $e^*$ must exist for the following reason. Let $s:=|S|$, $s\leq m$
  be the size of $S$, let $V_{e,i_1}\dots,V_{e,i_{s}}$ be components
  of $\forest\setminus \forest^*$ that contain some node of $S$, and
  let $V':=\bigcup_{j=1}^s V_{e,i_j}$. The number of edges of
  $\forest$ that connect two nodes in $V'$ is exactly $|V'|-s$ and
  because $\forest^*$ is also $S$-respecting, the number of edges of
  $\forest^*$ that connect two nodes in $V'$ is at most
  $|V'|-s$. Further, note that $e\in \forest\setminus\forest^*$
  contains at least one node $v\not\in V'$. Hence, since
  $\forest\setminus\set{e}$ has $m+1$ components and $\forest^*$ has
  only $m$ components, $\forest^*\setminus \forest$ must contain at
  least one edge $e^*$ that connects two components of
  $\forest\setminus\set{e}$, where at most one of those components is
  contained in $V'$. When choosing this edge $e^*$,
  $\forest\setminus\set{e}\cup\set{e^*}$ is an $S$-respecting spanning
  $m$-forest of $G$.

  By the assumptions of the theorem, we have
  $w(e)\leq \lambda\cdot w(e^*)$. To prove \eqref{appeq:MSFapprox}, it
  thus suffices to show that
  $W(\forest'\setminus \forest^*)\leq \lambda\cdot
  W(\forest^*\setminus \forest')$.
  We have $|\forest'\setminus \forest^*|=\gamma-1$ and thus, if the
  spanning forest $\forest'$ satisfies the conditions of the theorem,
  $W(\forest'\setminus \forest^*)\leq \lambda\cdot
  W(\forest^*\setminus \forest')$
  and \eqref{appeq:MSFapprox} follows from the induction hypothesis. We
  therefore need to show that $\forest'$ satisfies the conditions of
  the theorem. That is, we need to show that for every edge
  $e'\in \forest'\setminus\forest^*$ and for every edge
  $\hat{e}^*\in \forest^*\setminus \forest'$ such that
  $\forest'\setminus\set{e'}\cup\set{\hat{e}^*}$ is an $S$-respecting spanning
  $m$-forest, it holds that
  $w(\hat{e}^*)\geq w(e')/\lambda$.

  Let us therefore consider such a pair of edges $e'\in
  \forest'\setminus\forest^*$ and $\hat{e}^*\in \forest^*\setminus
  \forest'$ such that $\forest'\setminus\set{e'}\cup\set{\hat{e}^*}$ is an $S$-respecting spanning
  $m$-forest. We make a case distinction on whether $\hat{e}^*$
  connects two nodes of inside a single component of
  $\forest\setminus\set{e}$ or whether $\hat{e}^*$ connects two
  components of $\forest\setminus\set{e}$.
  \begin{itemize}
  \item Let us first assume that $\hat{e}^*$ connects two nodes $u$
    and $v$ inside a single component $V_{e,i}$ of
    $\forest\setminus\set{e}$. In this case, the edge $e'$ must
    connect two nodes $u'$ and $v'$ of the same component
    $V_{e,i}$. As a consequence, $\forest\setminus
    \set{e'}\cup\set{\hat{e}^*}$ is a spanning $m$-forest of $G$,
    which has the same component structure as $\forest'$. Hence, $\forest\setminus
    \set{e'}\cup\set{\hat{e}^*}$ is an $S$-respecting spanning
    $m$-forest and the assumptions of the theorem thus imply that
    $w(\hat{e}^*)\geq w(e')/\lambda$.
  \item Let us now assume that $\hat{e}^*$ connects two nodes $u$ an
    $v$ of different components $V_{e,i}$ and $V_{e,j}$ of
    $\forest\setminus\set{e}$. If both components $V_{e,i}$ and
    $V_{e,j}$ contain a node of $S$, the edge $e'$ must either
    connect two nodes in $V_{e,i}$ or two nodes in $V_{e,j}$ and we
    again have that $\forest\setminus
    \set{e'}\cup\set{\hat{e}^*}$ is an $S$-respecting spanning
    $m$-forest. Hence, the assumptions of the theorem again imply that
    $w(\hat{e}^*)\geq w(e')/\lambda$.

    It thus remains to consider the case where at most one of the
    components $V_{e,i}$ and $V_{e,j}$ contains a node of
    $S$. However, in this case, we can get an $S$-respecting
    $m$-forest by considering the tree
    $\forest\setminus\set{e}\cup\set{\hat{e}^*}$ and the assumptions
    of the theorem yield that $w(\hat{e}^*)\geq w(e)/\lambda$. This
    implies that $w(\hat{e}^*)\geq w(e')/\lambda$ because we assumed
    that $e$ is a maximum weight edge of $\forest\setminus\forest^*$
    and thus $w(e)\geq w(e')$.\qedhere
  \end{itemize}
\end{proof}

\begin{table}[H]
\centering
\begin{tabular}{|l|l|l|}
\hline
\textbf{Notation}                                 &\textbf{Definition}                                                    & \textbf{Page}                               \\ \hline
\begin{tabular}[c]{@{}l@{}}
$n$ \\ $k$ \\ \R \\ $r_i=(v_i,t_i)$ \\
$\pi^z_{\alg}$ \\ $r^z_0=(v^z,0)$ \\
$s^z$ \\ 
$\R^z_{\alg}$ \\ $\pi^z_{\alg}(i)$ \\ $\latency_{\alg}(\mu)$ \\
$\delay_{\alg}(r_i,r_j)$ \\ $\Delay_{\alg}(\pi^z_{\alg})$ \\
$\Delay_{\alg}$ \\ $d_G(u,v)$ \\ $H$ \\ $\mu(v)$ \\
$B$ \\ $r_{\src}(\mu)$ \\ $r_{\des}(\mu)$ \\
$e(\mu)$ \\ $\mu(e)$ \\ $T$ \\ $\forest_{\alg}$ \\
$F^z_{\alg}$ \\ $\Latency_{\alg}(F)$ \\
$\weight_{G}\big(e=(r_i,r_j)\big)$ \\ $\Weight_{G}(F)$ \\
$\R_D$ \\ $F(T')$ \\ $\forest_{\grd}$ \\ $(\mu,\mu')$ \\
$\forest_{mdf}$ \\ $E^{old}$ \\
$E^{new}$ \\ $E^{pot}$ \\ 
$E^{pot}(e)$ \\ $E^{old}(e)$ \\ $\Phi(F)$ \\
$\forest_{\min}$ \\ $M^{\uparrow}(T')$ \\
$M^{\downarrow}(T')$ \\ $\delta(T')$
\end{tabular}  & 

\begin{tabular}[c]{@{}l@{}}
number of pints/nodes/processors \\ number of servers \\ input requests \\ request $r_i$ that is invoked by node $v_i$ at time $t_i$ \\
$z$-th schedule as one of the $k$ resulted schedules by \alg \\ dummy request $z$ as the tail of $\pi^z_{\alg}$ \\
$z$-th server that serves all requests in $\pi^z_{\alg}$ \\ 
request set of $\pi^z_{\alg}$ \\ index of the request scheduled at the $i$-th position of $\pi^z_{\alg}$ \\ latency of message $\mu$ in an execution of \alg \\
cost incurred by \alg\ for scheduling $r_j$ as the successor of $r_i$ \\ total cost incurred by \alg\ for scheduling requests in $z$-th schedule \\
total cost incurred by \alg \\ weight of the shortest path between $u$ and $v$ on the input graph $G$ \\ directed version of $T$ that is changing during a \gnn\ execution \\ find-predecessor message sent by $v$ \\
complete graph on requests in $\R$ \\ corresponding request with message $\mu$ \\ predecessor request of $r_{\src}(\mu)$ \\
edge constructed by message $\mu$ \\ message that constructs the edge $e$ \\ input HST \\ resulted forest by \alg; also, set of edges of the forest \\
$z$-th TSP path of $\forest_{\alg}$; also, set of edges of the $z$-th TSP path \\ total cost of $F$ such that $F \subseteq \forest_{\alg}$ \\
weight of the shortest path between $v_i$ and $v_j$ on the input graph $G$ \\ total weight of $F$ w.r.t. measurements on the input graph $G$ \\
set of $k$ dummy requests; $\R_D \subseteq \R$ \\ subgraph of $F$ induced by the requests contained in $F$ and $T'$ \\ locality-based forest \\ gap \\
resulted forest by the transformation of $\forest_{\gnn}$ \\ set of edges removed throughout the transformation of  $\forest_{\gnn}$ \\
set of edges added throughout the transformation of $\forest_{\gnn}$ \\ $\left\{e \in \forest_{\gnn}:\big(\mu(e),\mu(e')\big) \ \text{is a gap for some} \ e' \in E^{old}\right\}$ \\
subset of $E^{pot}$ filtered out by $e \in E^{old}$ \\ subset of $E^{old}$ filtered out by $e \in E^{pot}$ \\ potential of $F$ \\
minimum weight $\R_D$-respecting spanning $k$-forest of $B$ \\ set of messages that leave $T'$ \\
set of messages that enter $T'$ \\ diameter of $T'$
\end{tabular} &

\begin{tabular}[c]{@{}l@{}}\pageref{no:noNodes} \\ \pageref{no:noServers} \\ \pageref{no:requests} \\ \pageref{no:singleRequest} \\
\pageref{no:schedule} \\ \pageref{no:dummy} \\
\pageref{no:server} \\
\pageref{no:requestPartition} \\ \pageref{no:requestPartition} \\ \pageref{no:latency} \\
\pageref{eq:delayCost} \\ \pageref{eq:componentTotalCost} \\ \pageref{eq:totalCost} \\ \pageref{no:shortestDist} \\ \pageref{no:H} \\ \pageref{no:messageSender} \\
\pageref{no:B} \\ \pageref{no:endpointsMSG} \\ \pageref{no:endpointsMSG} \\
\pageref{no:messageEdge} \\ \pageref{no:messageEdge} \\ \pageref{no:hstT} \\  \pageref{no:resultedForest} \\
\pageref{no:resultedForest} \\ \pageref{eq:totalLatencyF} \\
\pageref{no:weightEdge} \\ \pageref{eq:forestAlgLength} \\
\pageref{no:RD} \\ \pageref{no:subgraphTree} \\ \pageref{no:LBF} \\ \pageref{no:gap} \\
\pageref{no:transformedForest} \\ \pageref{no:newOldEdges} \\
\pageref{no:newOldEdges} \\ \pageref{no:potEdges} \\
\pageref{no:filterEPOT} \\ \pageref{no:filterEOLD} \\ \pageref{no:potential} \\
\pageref{no:minForest}  \\ \pageref{no:messagesUP} \\
\pageref{no:messagesDown} \\ \pageref{no:treeDiameter}
\end{tabular} \\ \hline

\end{tabular}
\caption{The essential notations used throughout the paper.}
\label{ta:notations}
\end{table}

\end{document}